\documentclass[aps,reprint,groupedaddress,nofootinbib,
               amsmath,amssymb,superscriptaddress,twocolumns,
               floatfix,balancelastpage,a4paper,accepted=2025-10-17]{quantumarticle}

\pdfoutput=1
\usepackage[utf8]{inputenc} 
\usepackage[T1]{fontenc}
\usepackage[a-1b]{pdfx} 
\usepackage{calc}
\usepackage{mathtools}
\usepackage{dsfont}
\usepackage{physics2}
\usephysicsmodule{ab, ab.braket, diagmat, ab.legacy}
\usepackage{todonotes}
\usepackage{xcolor}
\usepackage{placeins}
\usepackage{csquotes}
\usepackage{comment}
\usepackage{bm}
\usepackage{url}
\usepackage{etoolbox}
\usepackage{xspace}
\usepackage{xcolor}
\usepackage{xfrac}
\usepackage[expansion, protrusion]{microtype}
\usepackage{tikz}
\usetikzlibrary{calc, fit, shadows.blur}
\usepackage[compat=0.6]{yquant}
\usepackage{bm}
\usepackage{amsthm}
\usepackage[capitalise]{cleveref}
\usepackage{balance}
\usepackage{pdfrender}
\usepackage{bold-extra}
\usepackage{layouts}
\usepackage{censor}
\usepackage{makecell}
\usepackage{subcaption}
\usepackage{ragged2e}
\hypersetup{colorlinks=true,urlcolor=teal,
            linkcolor=teal,citecolor=teal}
\usepackage[numbers,sort&compress]{natbib}

\newboolean{anonymous}
\setboolean{anonymous}{false}
\ifbool{anonymous}{}{\renewcommand{\censor}[1]{#1}}
\ifbool{anonymous}{}{\renewcommand{\blackout}[1]{#1}}
\ifbool{anonymous}{}{}

\newtoggle{shownotes}
\newtoggle{showsketches}
\togglefalse{shownotes}
\togglefalse{showsketches}

\newcommand{\sketch}[1]{\iftoggle{showsketches}{\textcolor{gray}{#1}}{}}

\newcommand{\cconj}[1]{\overline{#1}}
\DeclareMathOperator*{\Cov}{Cov}
\DeclareMathOperator*{\Var}{Var}
\DeclareMathOperator*{\lb}{ld}

\DeclareMathOperator*{\diag}{diag}
\DeclareMathOperator*{\hypgeom}{HypG}

\newcommand{\mangle}{\beta}
\newcommand{\pangle}{\gamma}
\newcommand{\NP}{\textup{\textsc{\textbf{NP}}}\xspace}
\newcommand{\QAOA}{\textsc{qaoa}\xspace}
\newcommand{\NISQ}{\textsc{nisq}\xspace}
\newcommand{\SAT}{\textsc{sat}\xspace}
\newcommand{\ourQAOA}{non-iterative \QAOA}
\newcommand{\QUBO}{\textsc{qubo}\xspace}
\newcommand{\QUBOs}{\textsc{qubo}s\xspace}
\newcommand{\VQC}{\textsc{vqc}\xspace}
\newcommand{\VQCs}{\textsc{vqc}s\xspace}
\newcommand{\maxcut}{Max-Cut\xspace}
\newcommand{\ind}{\chi}
\newcommand{\kclique}{$k$-\textsc{Clique}\xspace}
\newcommand{\maxclique}{$\max$-\textsc{Clique}\xspace}
\newcommand{\qrfactoring}{$qr$-\textsc{Factoring}\xspace}
\newcommand{\Ccliques}{C_{\text{cliques}}}
\newcommand{\Cdhk}{C_{d_H = k}}

\newcommand{\sqboxsquare}[1]{%
    \begin{tikzpicture}
        \fill[#1] (0,0) rectangle (1.2ex, 1.2ex);
    \end{tikzpicture}
}

\newcommand{\sqbox}[1]{%
    \raisebox{0.15em}{\begin{tikzpicture}
        \fill[#1] (0,0.25em) rectangle (2.4ex, 0.5em);
    \end{tikzpicture}}
}

\definecolor{lfd1}{HTML}{FFFFFF} 
\definecolor{lfd2}{HTML}{E69F00}
\definecolor{lfd3}{HTML}{999999}
\definecolor{lfd4}{HTML}{009371}
\definecolor{lfdblue}{HTML}{1f78b4}

\tikzset{group/.style={rectangle, sharp corners,
                       draw=none, fill=#1, nearly transparent}}
\tikzset{compute/.style={fill=lfd2!40}}
\newcommand{\WIDTH}{2.55cm}
\newcommand{\HEIGHT}{1.65cm}
\tikzset{box/.style={rectangle, draw, fill=#1,
                     align=center, blur shadow},
         box/.default=white}
\tikzset{ibox/.style={box, minimum width=\WIDTH,
                     minimum height=\HEIGHT},
         box/.default=white}
\tikzset{varStyle/.style={draw, thin, solid, double}}
\tikzset{dataStyle/.style={draw, semithick, dashed}}
\tikzset{label/.style={fill=black!50, inner sep=0pt, text=white}}

\newcommand{\iterationnf}{\(\textcolor{lfd4}{\circlearrowright}\)}

\newcommand{\ie}{\emph{i.e.}\xspace}
\newcommand{\eg}{\emph{e.g.}\xspace}
\newcommand{\etal}{\emph{et al.}\xspace}

\newcommand{\doirepro}{\href{https://doi.org/10.5281/zenodo.13286242}{DOI-safe}\xspace}
\newcommand{\repro}{\href{https://github.com/lfd/qaoa-structural-approximation}{reproduction package}\xspace}

\AtEndPreamble{%
\newtheorem{theorem}{Theorem}
\newtheorem{lemma}{Lemma}
\newtheorem{definition}{Definition}
\newtheorem{remark}[theorem]{Remark}}

\begin{document}

\title{Out of the Loop: Structural Approximation of Optimisation Landscapes and non-Iterative Quantum Optimisation}

\author{Tom Krüger}
\orcid{0000-0002-1161-808X}
\affiliation{%
Technical University of Applied Sciences Regensburg,
Regensburg, Germany}
\affiliation{%
FI CODE,  Universität der Bundeswehr München,
Munich, Germany}
\email{tom.krueger@oth-regensburg.de}

\author{Wolfgang Mauerer}
\orcid{0000-0002-9765-8313}
\affiliation{%
Technical University of Applied Sciences Regensburg, Regensburg, Germany}
\affiliation{Siemens AG, Technology, Munich, Germany}
\email{wolfgang.mauerer@othr.de}

\maketitle
\begin{abstract}
The Quantum Approximate Optimisation Algorithm (\QAOA) is a widely
studied quantum-classical iterative heuristic for combinatorial
optimisation. While \QAOA targets problems in complexity class \NP, the classical optimisation procedure required in every
iteration is itself known to be \NP-hard. Still, advantage over classical approaches is suspected for certain scenarios, but nature 
and origin of its computational power are not yet satisfactorily understood.

By introducing means of efficiently and accurately approximating the \QAOA optimisation
landscape from solution space structures, we  derive a new algorithmic variant of unit-depth \QAOA for
two-level Hamiltonians (including all problems in \NP): 
Instead of performing an iterative quantum-classical computation for each input 
instance, our non-iterative method is based on a quantum circuit that 
is instance-independent, but problem-specific. It matches or outperforms unit-depth \QAOA for key combinatorial problems, despite 
reduced computational effort.

Our approach is based on proving a long-standing conjecture regarding
instance-independent structures in \QAOA. By ensuring generality, we link existing empirical observations on \QAOA parameter clustering to established approaches in 
theoretical computer science, and provide a sound foundation for understanding the link between
structural properties of solution spaces and quantum optimisation.
\end{abstract}

\section{Introduction}
With the advent of early commercially available quantum computers, interest in the field has spilled from academia to early industrial adopters, and
the hope for possible advantage or even supremacy
is manifest~\cite{Pirnay:2024,Arute:2019,montanez2024towards}. However, challenges arise not only from deficiencies of noisy, 
intermediate-scale quantum (\NISQ) hardware~\cite{Bharti:2022,greiwe:23:imperfections,Thelen:2024,blekos_review_2024}, but quantum algorithmic theory in general lags behind the classical case. While fundamental complexity-theoretic boundaries have long been established~\cite{Adleman:1997,Bernstein:1997,Fortnow:1999}, a more precise understanding of concrete algorithmic building blocks and how to construct them is required~\cite{Aaronson:2015,Montanaro:2016,Shao:2019}.

How to systematically construct quantum algorithms is a multi-faceted, highly non-trivial task that remains essentially unsolved. Instead, heuristics like variational quantum circuits (\VQC)~\cite{Cerezo:2021} have emerged as popular alternatives. Originally intended to extract useful computational power from \NISQ devices despite their limitations, they may also prove relevant as resource-efficient primitives in the post-\NISQ era. They are centred around an iterative quantum-classical process that learns a parameterised form of a quantum circuit such that sampling a produced quantum state obtains a valid solution with high probability. Depending on the variational ansatz, this approach defers considerable aspects (\ie, finding an
appropriate quantum circuit that eventually implements the optimisation routine) of the task to classical components~\cite{PellowJarman:2024,FernndezPends:2022,zhou2020quantum}.

Structured forms of \VQCs, in particular the quantum approximate 
optimisation algorithm (\QAOA), enjoy popularity in prototypical 
applications~\cite{bayerstadler:21:}. \QAOA is a specific \VQC for 
combinatorial optimisation. Similar to quantum annealing, to which \QAOA 
is closely related~\cite{pelofske2024short,Bharti:2022}, it 
provides a strict framework on the quantum side. Given a 
specific problem, users only need to find a fitting problem 
representation, choose a suitable circuit depth, and pick a classical 
optimisation method. Devising such problem representations 
is not foreign to traditional computer science: In 
particular, techniques to reduce computational problems
specified using an apt formalism into
representations that are suitable for \QAOA~\cite{schmidbauer:24:reductions,cipra2000ising,lucas2014ising} are curricular knowledge~\cite{Arora:2009}.

\begin{figure*}[htbp]
    \centering
    \includegraphics{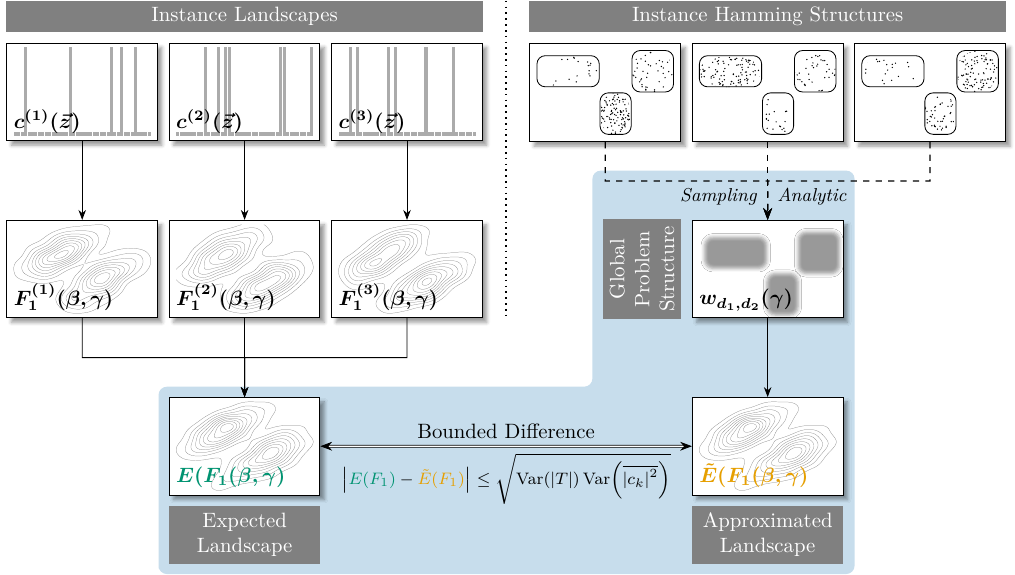}
    \caption{Overview of our main foundational contributions. \emph{Left:} Previously conjectured parameter clustering in \QAOA optimisation landscapes  \(F_{1}^{(i)}(\mangle, \pangle)\) (derived
    from a combinatorial optimisation objective function \(c^{(i)}(\vec{z})\)
    of a decision problem instance) suggests that shared macroscopic similarities exist between instances \(i\). Averaging over these reveals that the expected landscape \(E(F_{1}(\mangle, \pangle))\) is a common structure at the problem-global level. 
    \emph{Right:} Shared macroscopic features exist across all instances, based on structural properties manifest in the solution spaces. Aggregation (which may be possible analytically, but can always be performed using empirical sampling) leads to a macroscopic description of a 
    problem-global solution space, from which the expected optimisation landscape 
    can be efficiently approximated as \(\tilde{E}(F_{1}(\mangle, \pangle))\). Most importantly, we prove that the approximation has a bounded difference to the underlying exact quantity. The approximation approach (ingredients indicated
    by a blue background)  and its consequences are subject of this paper.}\label{fig:overview}
\end{figure*}

The most essential ingredient of combinatorial optimisation is, of course, 
the optimisation landscape. This paper is centred around a new theorem, detailed 
in \cref{sec:approx_thm}, that, for two level Hamiltonians (including all problems 
in \NP), allows us to approximate the expected one layer \QAOA optimisation landscape of a 
problem from existing solution space structures. Informally speaking, 
\QAOA is concerned with multiple objects: A problem (\eg, can a graph be 
partitioned into two halves such that only a certain amount of edges must be 
cut?), an instance (\eg, a specific graph), a solution space (\eg, lists of edges to cut),
and parameters that define an appropriate quantum circuit 
specific for each \emph{instance} to obtain solutions from sampling the circuit output.

Hamming distances between solutions offer an intuitive handle to discover and describe local structures in solution space. Our approach continues
a line of research~\cite{montanez2024towards,diez2024connection,
streif2019comparison,bravyi2020obstacles} using such distance information, particularly the connection to state amplitudes and their interference,
to obtain insights about the inner workings of \QAOA. By using stochastic 
methods, we establish novel means of reasoning about \QAOA landscapes on an \emph{instance 
independent}, but \emph{problem specific} level, and show that 
instance-specific quantum circuits can be replaced by one single 
problem-generic alternative, while obeying a strictly bounded  
maximal difference between approximated and exact ingredients of
the overall combinatorial optimisation process.
Our result contributes to both, foundational understanding
of \QAOA and practical implementations.

As for the fundamental aspects, we provide a rigorous mathematical proof of a long-standing empirically motivated conjecture~\cite{brandao2018fixed,streif2020training,sack2021quantum,galda2similarity,Sud:2024} that instance independent structures of a problem form the global structure of the \QAOA optimisation landscape. We combine methods from physics and computer science to ascertain this structural insight, and can use either efficient sampling techniques or an analytical approach to obtain an accurate approximation of the \QAOA optimisation landscape. In particular, we can separate the instance sampling from specific \QAOA parameters. These aspects of our contribution are visually summarised in Fig.~\ref{fig:overview}.
\newbox\varbox\newbox\databox
\setbox\varbox\hbox{\raisebox{0.15em}{\tikz{\path[varStyle] (0,0) -- (0.5,0);}}}
\setbox\databox\hbox{\raisebox{0.15em}{\tikz{\path[dataStyle] (0,0) -- (0.5,0);}}}

\begin{figure*}[htbp]
    \includegraphics{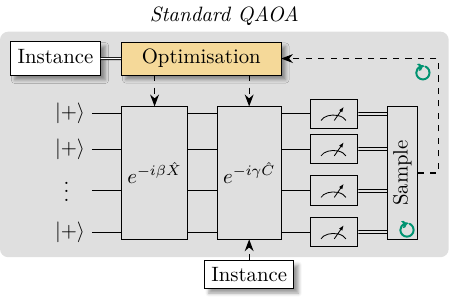}
    \hfill
    \includegraphics{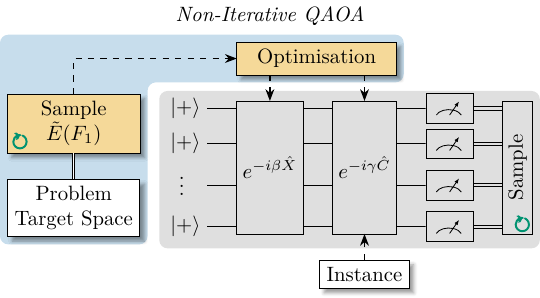}
    \caption{Overview of our main practical contributions (blue background \sqboxsquare{lfdblue, nearly transparent} indicates \emph{problem-global},
    grey background \sqboxsquare{gray!25} \emph{instance-specific} components; yellow boxes \sqboxsquare{compute} mark classical computation. Dashed lines \usebox\databox{} indicate transfer of classical data, and double lines \usebox\varbox{} denote querying a resource).
    \emph{Left:} Standard \QAOA that iteratively (symbolised by \iterationnf) determines optimal parameters (\(\mangle, \pangle\)) for every instance by repeatedly sampling a structured quantum circuit.
    \emph{Right:} Two-phase \QAOA approach introduced in this paper that first
    determines optimal parameters (\(\mangle, \pangle\)) by sampling \(\tilde{E}(F_{1}(\mangle, \pangle))\) from the problem-global target space, and then uses the derived instance-independent optimal parameters (\(\mangle, \pangle\)) to obtain instance-specific solutions by sampling from a quantum circuit that is constant for each
    instance \(\hat{C}\). As we show, \(\tilde{E}(F_{1})\) has a strictly bounded difference to the expectation value of 
    \(F_{1}\), the core quantity of interest in \QAOA.
    }\label{fig:application_overview}
\end{figure*}

In terms of practical impact, we provide a unit-depth non-iterative \QAOA version
for two-level problem Hamiltonians. Such two-level Hamiltonians are powerful enough
to efficiently encode all problems in \NP, allowing for the construction of \QAOA circutis
to solve said problems. Thus, for \NP-Problems, we can replace the standard 
iterative \QAOA heuristic comprising an interplay of quantum and classical 
components (that, notably, requires solving an \NP-hard parameter 
optimisation problem~\cite{bittel_training_2021} individually for 
\emph{every} instance) with a two-phase, non-iterative algorithm that first approximates 
the instance-independent, but problem-specific expected landscape 
followed by sampling a fixed quantum circuit, as 
illustrated in \cref{fig:application_overview}. While omitting
the outer loop has, in various ways, been suggested in prior work based on the 
aforementioned empirical observations~\cite{brandao2018fixed,streif2020training},
our approach places the idea on a sound theoretical basis.
It also introduces a new method of devising classical optimal parameters,
and equips the method with a quantitative quality criterion. We demonstrate that our
simplified variant is identical or better to standard \QAOA for several 
seminal subject problems in terms of successfully finding solutions
to combinatorial optimisation problems.

By showing that the required structures exist for all 
problems in \NP, we establish generality of our results.
Our approach therefore opens up the discussion on parameter 
clustering and optimisation landscape similarities to the vast body 
of established results in theoretical computer science on structural 
properties of problem solutions. 

The rest of this paper is organised as follows: 
\Cref{sec:foundations} first presents related work, supporting and 
motivating our ideas, and then establishes precise 
terminology for \QAOA fundamentals, particularly given the many subtly 
different variants in current use. \Cref{sec:opt_landscape} covers 
the \QAOA optimisation landscape and prepares the necessary groundwork 
for \cref{sec:approx_thm}, where the main approximation theorem is 
presented and proved. With the new theorem at hand, we discuss a list 
of five examples building on each other in \cref{sec:application}. This 
illustrates how our theorem can be used to analyse synthetic and realistic 
computational problems. The practical application of our insights is
subject of \cref{sec:utility}, where we introduce a novel non-iterative variant of \QAOA based on our insights that matches or exceeds the performance of the standard 
heuristic for a set of subject problems, yet at substantially reduced computational cost. After setting our contributions
in context with the current state of the art, and discussing potential 
further uses and future improvements in \cref{sec:discussion}, we
conclude in \cref{sec:conclusion}.

\section{Foundations}\label{sec:foundations}
\subsection{Context and Related Work}\label{sec:relatedwork}

\sketch{
Parameter concentration:
\begin{itemize}
    \item Analytical definition of parameter concentration / target state projector \cite{akshay2021parameter} \citeyear{akshay2021parameter}
    \item Instance clustering \maxcut, and instance transferability / cost landscapes (very similar to ours) of sub graphs\cite{galda2similarity} \citeyear{galda2similarity}
    \item Landscape values concentrate for fixed parameter / conjecture landscapes are instance invariant \cite{brandao2018fixed} \citeyear{brandao2018fixed}
    \item Observed parameter concentration / Training without access to a QPU \cite{streif2020training} \citeyear{streif2020training} 
\end{itemize}
classical work on solution space structures:
\begin{itemize}
    \item community structures in industrial \SAT instances: \cite{AnsoteguiGL12} \citeyear{AnsoteguiGL12}
    \item even random \SAT instances exhibit solution structures, they also used Hamming distances to capture structures: \cite{PariLYQ04} \citeyear{PariLYQ04}
    \item clustering: realistic instances often models physical or social structures / interaction likelihood depends on distance metric \cite{hogg1996refining} \citeyear{hogg1996refining}
    \item clustering of \SAT solutions / in lower alpha regions exponentially many small clusters\cite{achlioptas2011solution} \citeyear{achlioptas2011solution}
\end{itemize}
Hamming distance:
\begin{itemize}
    \item grouping amplitude contributions by Hamming distance (Ising Hamiltonian) \cite{montanez2024towards} \citeyear{montanez2024towards}
    \item Eigenspace amplitudes Bolzmann distribution ($p = 1$)\cite{diez2024connection} \citeyear{diez2024connection}
    \item Comparing \QAOA with QA and SA / one target state (p = 1) \cite{streif2019comparison} \citeyear{streif2019comparison}
    \item Local Hamiltonians highly entangled ground states results extends to low energy states in presence of symmetry\cite{bravyi2020obstacles} \citeyear{bravyi2020obstacles}
    \item the alternating operator ansatz with mixers that independently work on Hamming distance "manifolds"
\end{itemize}
}

\QAOA and variants of the algorithm have been subject of intensive 
research~\cite{galda2similarity,fernandez-pendas_study_2022, 
bravyi2020obstacles, Hadfield:2019, akshay2021parameter, 
lotshaw_empirical_2021, Weidenfeller:2022, Alam:2022, Harrigan:2021, 
Herrman:2021, Lee:2021, Medvidovi:2021, Willsch:2020, Wang:2020, 
Bengtsson:2020, Wecker:2016, Jiang:2017, Morales:2020, Lechner:2020, 
Pan:2022, Guerreschi:2019, Gogeissl:2024, Pagano:2020, 
sack2021quantum,Egger:2021,Vijendran:2024, Thelen:2024, Safi:2023, 
Wintersperger:2022, Dupont:2022, schmidbauer:24:reductions, McGeoch:2023, 
Sud:2024, Awasthi:2023,Baertschi:2020, Tate:2023,Fingar:2024,Singhal:2024, 
FernndezPends:2022,PellowJarman:2024,StilckFrana:2021,Shaydulin:2024,Ozaeta:2022}, as recently reviewed 
by Blekos~\etal~\cite{blekos_review_2024} or 
Zhou~\etal~\cite{zhou2020quantum}. Apart from improving the understanding of 
the heuristic construction, modifications to the structure of the quantum 
circuit itself (\eg, Refs.~\cite{Zhang:2017,Wang:2020,Baertschi:2020,
bravyi2020obstacles,zhou2020quantum}) aim at improving
performance especially in \NISQ scenarios. Likewise, changes to the classical 
optimisation procedure (\eg, 
Refs.~\cite{Egger:2021,Awasthi:2023,Tate:2023,Vijendran:2024,Sud:2024,
montanez2024towards,Fingar:2024,streif2020training}) have been proposed.
Given that limitations of \NISQ hardware restrict programs to shallow 
circuits, many analytical and practical studies focus on unit-depth 
\QAOA~\cite{farhi2014quantum,bravyi2020obstacles,galda2similarity,blekos_review_2024,fernandez-pendas_study_2022,lotshaw_empirical_2021,Jain:2022,sack2021quantum,Egger:2021,Vijendran:2024,Hadfield:2019}; our considerations are also based
on this commonly employed scenario. Perhaps fuelled by the possibility of empirically investigating early-stage deployments, \QAOA has been applied to a considerable variety of practical problems; see Bayerstadler~\etal~\cite{bayerstadler:21:} for a review.

Our work is particularly motivated by efforts that observe concentrations of optimal \QAOA parameters across instances, which has received substantial consideration in the literature. In \citeyear{streif2020training}, \citeauthor{streif2020training}~\cite{streif2020training} observed a clustering of optimal parameters when training \QAOA circuits for random \maxcut instances, which inspired them to propose
training \QAOA without quantum hardware access. \Citeauthor{wybo2024missing}, recently used these methods to analyse fixed parameter \QAOA landscapes for Ising formulations of the maximum cut and the maximum independent set problems \cite{wybo2024missing}. Here, they also observed a promising performance of \QAOA with fixed parameters. \citeauthor{sack2021quantum}~\cite{sack2021quantum} reported similar parameter concentrations in \citeyear{sack2021quantum} also for \maxcut instances. They provided a theoretical definition of the effect based on the closeness of optimal sets in parameter space. Their highlighting the importance of picking good initialisation values for \QAOA parameters evolved into \emph{warm start \QAOA}: Here, the focus lies on finding optimal initialisation parameters~\cite{Egger:2021}; this is obviously tightly linked to optimal parameters concentrating in certain regions. Predicting the location of those clusters would benefit warm start efforts. Following that, \citeauthor{akshay2021parameter} published a thorough analytic analysis, providing a new definition and concrete analytic insights \cite{akshay2021parameter}. They, for instance, showed that \QAOA circuit parameters concentrate as an inverse polynomial in problem size. \citeauthor{galda2similarity}~\cite{galda2similarity} demonstrated the transferability of \QAOA parameters between different \maxcut instances in \citeyear{galda2similarity}, and explained their findings with local graph properties. They also investigated optimisation landscapes of certain sub-graphs and observed stark similarities between different sub-graphs. This overlaps with earlier observations by \citeauthor{brandao2018fixed} in  \citeyear{brandao2018fixed}~\cite{brandao2018fixed} that fixed parameters of the optimisation landscape concentrate around certain values for different instances. All these findings suggest the existence of higher level problem structures shared between instances that significantly influence the shape of the underlying optimisation landscapes. 

Classical theoretical computer science has studied structural properties of problems for decades. Among the most fundamental structural observations are phase transitions in constraint satisfaction problems: At a certain point, problem instances transition from under-constrained to over-constrained. The parameter correlated with this effect depends on the problem: For the seminal problem of Boolean satisfiability (\SAT), it is known to be the ratio of numbers of clauses to the number of variables in an instance \cite{Monasson:1999}. For graph colouring, the connectivity of the underlying graph is the relevant determinant. Finding a solution is especially hard at the phase transition. While the phase transitions provide a high level description of problem hardness, more complex structural properties are known. \citeauthor{hogg1996refining} argued in \citeyear{hogg1996refining}~\cite{hogg1996refining} that real-world applications model interactions between physical or social entities. Thus, interaction likelihood strongly depends on the distance between entities based on some appropriate distance measure, leading to recurring local structures. In \citeyear{PariLYQ04}, \citeauthor{PariLYQ04}~\cite{PariLYQ04} showed that even trivially sampled random \SAT instances do possess structured solution spaces. The probability of a potential solution satisfying the \SAT formula depends on the Hamming distance to other solutions. A few years later in \citeyear{achlioptas2011solution}, \citeauthor{achlioptas2011solution}~\cite{achlioptas2011solution} discovered a clustering of the solution space of \SAT instances. They further proved that under-constrained \SAT formulas have exponentially many small clusters in their solution space. The presence of community structures in the solution space of industrial \SAT instances was demonstrated in \citeyear{AnsoteguiGL12} by \citeauthor{AnsoteguiGL12}~\cite{AnsoteguiGL12}.

\subsection{Quantum Approximate Optimisation Algorithm}
More often than not, the \emph{Quantum Approximate Optimisation Algorithm} (\QAOA) is associated with the \emph{quadratic unconstrained binary optimisation} (\QUBO) problem, which is NP-complete in its decision form and relatively straight forward to solve with \QAOA. A \QUBO problem is defined by Boolean quadratic formula of the form $\sum_{i \neq j} a_{i,j} x_i x_j + \sum_{i} a_i x_i$, where $a_{i,j},  a_i \in \mathds{R}$ are real valued weights of the Boolean variables $x_i \in \mathds{F}_2$, with $\mathds{F}_2 \coloneqq \ab\{0, 1\}$. The goal is to find a variable assignment to maximise this \QUBO formula. This can be easily formulated as a ground state problem of an Ising Hamiltonian $\sum_{i \neq j} -J_{i,j} \sigma_i^z \sigma_j^z - \sum_i h_j \sigma_i^z$, with $J_{i,j}, h_j \in \mathds{R}$ and $\sigma_i^z$ being the Pauli-Z operator on the $i$-th qubit.
Therefore, it seems reasonable to view \QAOA as a dedicated \QUBO solver. This approach is analogous to how \SAT solvers are employed in classical systems. There is a rich community of \SAT experts working on newer and better solvers, while the users on the other side can rely on the interface of abstract \SAT formulas to solve their concrete use cases without them needing to dive deep into the intricacies of Boolean satisfiability. Similarly, if we look at \QAOA as just another \QUBO solver, this takes the majority of \emph{quantum} out of quantum computing as the \QUBO formalism serves as a classical interface. This is arguably a major reason why \QUBOs have been a welcoming entry point to quantum computing, particularly for researchers 
who do not feel the need to understand details of the computational process~\cite{Kochenberger:2014,lucas2014ising}. Accompanying that, promising methods of solving \QUBO problems with means other than quantum computing have enjoyed a certain amount of attention~\cite{Aramon:2019,Henke:2023,Schoenberger:2023,Alom:2017,Seker:2022}.

While our results apply to \QUBO problems, our considerations 
are not restricted to this scenario, but
consider \QAOA as proposed by \citeauthor{farhi2014quantum} to solve general combinatorial optimisation problems \cite{farhi2014quantum}, of which \QUBOs only form a restricted subset.

\begin{definition}
  \label{def:combinatorial_optimisation}
  Let $z \in \mathds{F}_2^n$ be a $n$-bit binary string. Further $\{c_\alpha\}_{\alpha = 1}^m$ shall be a set of Boolean clauses with $c_\alpha (z) = 1$ iff $z$ satisfies $c_\alpha$ and $c_\alpha (z) = 0$ otherwise. Maximising
  \begin{equation}
    \label{eq:classical_constraint_cost_func}
    c(z) = \sum_{\alpha = 1}^m c_\alpha (z)
  \end{equation}
  is known as the Boolean constraint optimisation problem.
\end{definition}

We can bring \cref{def:combinatorial_optimisation} to the quantum world by defining a corresponding basis state vector $\ket|z>$ for each bit string $z \in \mathds{F}_2^n$. The obvious choice for $\ket|z>$ is the computational basis vector $\ket|z> \in \{\ket|0>, \ket|1>\}^{\otimes n}$ encoded by $z$. Then we map the clause values $c_\alpha(z)$ to eigenvalues of quantum operators $C_\alpha$ representing the clauses. For a one to one mapping, we get a projector $C_\alpha$ per clause $c_\alpha$ that projects onto the subspace spanned by all states representing satisfying assignments of $c_\alpha$. With Hermitian operators being closed under addition, we have that $C = \sum_{\alpha = 1}^m C_\alpha$ is itself a Hermitian operator. This allows us to define the Hamiltonian time evolution operator $e^{-i \pangle C}$. Hamiltonian $C$ describes the energy of a quantum system, and energy levels are eigenvalues of $C$. The solution $z$ of \cref{def:combinatorial_optimisation} maximises \cref{eq:classical_constraint_cost_func} and thus the corresponding eigenstate $\ket|z>$ of $C$ has eigenvalue $\lambda_\text{max} = \text{max} \:\: \sigma (C)$. Here, $\sigma\ab(C)$ is the set of eigenvalues of $C$. From the Hamiltonian point of view, solving the combinatorial optimisation problem is equivalent to finding a state with maximal energy of a system described by the Hamiltonian $C$. \citeauthor{farhi2014quantum} came up with \QAOA by trotterising an interpolated Hamiltonian time evolution from an easy to prepare maximal energy state of a system described by a simple Hamiltonian to the state of maximal energy of the system described by the constraint Hamiltonian $C$, in which the problem structure of \cref{def:combinatorial_optimisation} is encoded. 

\begin{definition}[\QAOA circuit as in Ref.~\cite{farhi2014quantum}]
  \label{def:qaoa_circuit}
    We consider a constraint Hamiltonian $C$ implementing the constraint cost function \cref{eq:classical_constraint_cost_func} for $\bm{z} \in \mathds{F}_2^n$ and a mixer Hamiltonian $X = \sum_{j = 1}^n \sigma_j^x$, with $\sigma_i^x = \mathds{1}^{\otimes i - 1} \otimes \sigma^x \otimes \mathds{1}^{n - i} $. Then, the \QAOA circuit
  \begin{equation}
    \label{eq:qaoa_circuit}
      U_p \ab(\bm{\mangle}, \bm{\pangle}) = e^{-i \mangle_p X} e^{-i \pangle_p C} \cdots e^{-i \mangle_1 X} e^{-i \pangle_1 C}
  \end{equation}
  produces the state
  \begin{equation}
    \label{eq:qaoa_state}
    \ket|\bm{\mangle}, \bm{\pangle}>_p = U_p \ab(\bm{\mangle}, \bm{\pangle}) \ket|+>^{\otimes n}
  \end{equation}
  with real angles $\bm{\mangle}, \bm{\pangle} \in \mathds{R}^p$, $p \in \mathds{N}$.
\end{definition}

\Cref{def:qaoa_circuit} actually defines a parameterised family of circuits $\ab\{U_p \ab(\bm{\mangle}, \bm{\pangle})\}_{\bm{\mangle}, \bm{\pangle}, p}$. For the Hamiltonian implementation of \cref{eq:classical_constraint_cost_func}, the evaluation $c(\bm{z})$ can be performed by calculating the expectation value $\braket<\bm{z}|C|\bm{z}>$. \citeauthor{farhi2014quantum} showed that $\lim_{p \to \infty} \bra<\bm{\mangle}, \bm{\pangle}|_p C \ket|\bm{\mangle}, \bm{\pangle}>_p = \max_{\bm{z}} c(\bm{z})$. This lets one define an algorithm to approximate the combinatorial optimisation problem with the circuit defined in \cref{eq:qaoa_circuit}.

\begin{definition}[\QAOA]
  \label{def:qaoa}
    Consider a combinatorial optimisation problem with constraint cost function $c(\bm{z})$ and $\bm{z} \in \mathds{F}_2^n$. For a fixed $p \in \mathds{N}$, choose a set of angles $\bm{\mangle}, \bm{\pangle} \in \mathds{R}^p$ that maximises the \QAOA cost function 
  \begin{equation}
      F_p(\bm{\mangle}, \bm{\pangle}) = \bra<\bm{\mangle}, \bm{\pangle}|_p C \ket|\bm{\mangle}, \bm{\pangle}>_p
  \end{equation}
    Then construct the circuit $U_p\ab(\bm{\mangle}, \bm{\pangle})$, with which the state $\ket|\bm{\mangle}, \bm{\pangle}>_p$ will be prepared and measured in the computational basis to produce a binary string $\bm{z} \in \mathds{F}_2^n$. Repeat this sampling step $m$ times with the same circuit to get a binary string that is close to $\max_{\bm{z}} c(\bm{z})$ with high probability.
\end{definition}

Strictly speaking, \cref{def:qaoa} defines a heuristic and not an algorithm. The process of finding optimal angles $\bm{\mangle}, \bm{\pangle} \in \mathds{R}^p$ is not further specified and open to interpretation. As shown by \citeauthor{farhi2014quantum}~\cite{farhi2014quantum}, $F_p$ can be simplified for specific problems---they 
consider \maxcut on 3-regular graphs---which allows for finding an efficient classical evaluation of $F_p$. It would also be feasible to evaluate $F_p$ on quantum hardware. The possible parameter optimisation methods are plentiful~\cite{FernndezPends:2022,PellowJarman:2024,blekos_review_2024}, and their
impact is subject to ongoing research.

Even for $p=1$, the parameter optimisation problem of $F_1 \ab(\mangle, \pangle)$ is \NP-hard~\cite{bittel_training_2021}. We restrict our considerations to a single layer, as is common practice~\cite{farhi2014quantum,bravyi2020obstacles,galda2similarity,blekos_review_2024,fernandez-pendas_study_2022,lotshaw_empirical_2021,Jain:2022,sack2021quantum,Egger:2021,Vijendran:2024,Hadfield:2019}. For the sake of simplicity, we also focus on constraint Hamiltonians with two-level eigenspectra. This avoids some effort in notation for the following 
theorems, but still allows us to solve problems in \NP. Note that the structural properties we prove below also exist for constraint Hamiltonians $C$ with more than two distinct eigenvalues. Also, two-level constraint Hamiltonians include all decision problems with classical proofs $\bm{z} \in \mathds{F}_2^n$, most notably the complete class of \NP. Despite the restriction to $\ab|\sigma(C)| = 2$ and \(p=1\), our 
setting therefore covers a large body of non-trivial, interesting computational 
problems. 

Before we proceed with our analysis, let us fix terminology
regarding \QAOA. The unitary gates $e^{-i \pangle_i C}$ defined by the Hamiltonian time evolution of $C$ are usually called phase separation gates or just phase separators.  As becomes clear in Eq.~\ref{eq:psep} 
below, it separates the solution space from the search space by a 
complex phase, and hence earns its name.
The term $e^{-i \mangle_i X}$ is usually called mixer, and $X$ is
the mixer Hamiltonian. Together, a phase separator and mixer pair forms a layer $e^{-i \mangle_i X} e^{-i \pangle_i C}$. In this case we speak of the $i$-th layer of a $p$-layer \QAOA circuit $U_p(\bm{\mangle}, \bm{\pangle})$.

\section{The Optimisation Landscape}\label{sec:opt_landscape}

Now we want to take a look at the optimisation landscape induced by the decision problem derived from \cref{def:combinatorial_optimisation}, which asks whether there exists an assignment $\bm{z} \in \mathds{F}_2^n$ satisfying all clauses---or, equivalently, if $c(\bm{z}) = \prod_{\alpha} c_\alpha(\bm{z}) = 1$. This directly translates to the Hamiltonian implementation $C = \prod_\alpha C_\alpha$ which, as a product of projectors, remains a projector itself. Thus, the eigenspectrum of $C$ is $\sigma(C) = \{0,1\}$. Since $C$ is Hermitian, there exists a unitary diagonalisation $C = U \begin{psmallmatrix} \mathds{1} &  \\ & \bm{0}  \end{psmallmatrix} U^\dag$. Note that we can collect all eigenvalues $\lambda = 1$ in the upper left block of the diagonal matrix by simply applying a permutation operator, which we can subsume into the unitaries $U$ and $U^\dagger$. Using the power series expansion of the exponential function, we see that the exponentiation can be passed through the diagonalisation. Let $H$ be a diagonalisable operator with $H = U D U^\dagger$ with $D = \diag (d_1, \dots, d_n)$ and unitary $U$. Then using the power series expansion of the exponential function, we get $e^H = \sum_{k = 0}^\infty \frac{1}{k!} (U D U^\dagger)^k$. Since $U$ is a unitary operator, the inner contributions $U^\dagger U=\bm{1}$ vanish in the product, and $(U D U^\dagger)^k = U D^k U^\dagger$. It follows that 
\begin{equation}
e^H = \sum_{k = 0}^\infty \frac{1}{k!} U D^k U = U \ab(\sum_{k = 0}^\infty \frac{1}{k!} D^k) U^\dagger. 
\end{equation}
Since $D$ is diagonal, we find that
\begin{equation}
\begin{aligned}
    \sum_{k = 0}^\infty D^k &= \diag \ab(\sum_{k = 0}^\infty d_1^k, \dots, \sum_{k = 0}^\infty d_n^k) \\
                            &= \diag(e^{d_1}, \dots, e^{d_n}).
\end{aligned}
\end{equation}
Consequently, we can express the phase separator by \(e^{-i \pangle C} = U \begin{psmallmatrix} e^{-i \pangle} \mathds{1} & \\ & \mathds{1} \end{psmallmatrix} U^\dag\).
Let $\ket|z>$ be a computational basis state. Then
\(\begin{psmallmatrix} e^{-i \pangle} \mathds{1} & \\ & \mathds{1} \end{psmallmatrix} U^\dag \ket|z> = e^{-i \pangle} U^\dag \ket|z> = U^\dag e^{-i \pangle} \ket|z>\) iff $\ket|z>$ has eigenvalue 1 under $C$, and \(\begin{psmallmatrix} e^{-i \pangle} \mathds{1} & \\ & \mathds{1} \end{psmallmatrix} U^\dag \ket|z> = U^\dag \ket|z>  \) otherwise.
If we apply the full phase separator gate on the state, we obtain \(e^{-i \pangle C} \ket|z> = U \begin{psmallmatrix} e^{-i \pangle} \mathds{1} & \\ & \mathds{1} \end{psmallmatrix} U^\dag \ket|z> = UU^\dag e^{i\pangle}\ket|\phi> = e^{i\pangle}\ket|z>\) iff $\ket|z>$ has eigenvalue 1 under $C$, and \(UU^\dag \ket|z> = \ket|z>\) otherwise.
We conclude that $e^{-i \pangle C}$ adds a global phase to a computational basis state $\ket|z>$ if and only if $\ket|z>$ has the eigenvalue 1 under C, and leaves the state invariant
otherwise. Applying $e^{-i \pangle C}$ to an arbitrary superposition of computational basis states $\ket|\psi> = \sum_{z = 0}^{2^n -1} \omega_z \ket|z>$ can be characterised by

\begin{equation}
  e^{-i \pangle C} \ket|\psi> = e^{-i \pangle} \sum_{\ket|z> \in T} \omega_z \ket|z> + \sum_{\ket|z> \notin T} \omega_z \ket|z>.\label{eq:psep}
\end{equation}
Here, $T = \ab\{\ket|z> \mid c(z) = 1\}$ is the target space of $C$, which directly maps to the solution space of \cref{eq:classical_constraint_cost_func}. At this point, the amplitude symmetry is broken in the \QAOA circuit, which allows
for establishing interference effects whose pattern
can be controlled by $\pangle$ and $\mangle$. This, eventually, allows us
to benefit from quantum effects in the
computational process.

The angle parameters $\mangle$ also have an interesting effect on state amplitudes. As we will shortly show in detail, solely the Hamming distances between states and the target/non-target partition of the state space suffice to describe the inner workings of \QAOA circuits for decision problems. However, we need to analyse the 
effect of mixer layers before we can commence to proving this
statement.

\begin{lemma}[Projector version of \cite{diez2024connection}]
  \label{thm:mixer_amplitudes_f}
    Let $X = \sum_{j = 1}^n \sigma_j^x$ be the $n$-qubit mixer Hamiltonian. The effect of $e^{-i \mangle X}$ on an arbitrary basis state $\ket|z> \in \mathds{F}_2^n$ can be characterised by
  \begin{equation}
    e^{-i \mangle X} \ket|z> = \sum_{k \in \mathds{F}_2^n} f(\mangle, z, k) \ket|k>
  \end{equation}
  where $f$ is defined as 
  \begin{equation}
    \label{eq:f_with_states}
    f(\mangle, z, k) \coloneqq (\cos \mangle)^{n - d_H (z,k)} (-i \sin \mangle)^{d_H (z,k)},
  \end{equation}
  and $d_H(z,k)$ is the Hamming distance of $z$ and $k$.
\end{lemma}

\begin{proof}
  Given that $X$ is a 1-local Hamiltonian, we can analyse its time evolution by only considering its effect on single qubits. Let us look at $n=1$ first: In this case, $X = \sigma^x$, and $e^{-i \mangle X} \ket|z> = \cos \mangle \ket|z> - i \sin \mangle \ket|z \oplus 1>$ for $z \in \mathds{F}_2$, which performs a bit flip with probability $\ab(\sin \mangle)^2$. For arbitrary $n \in \mathds{N}$, this generalises to
\begin{equation}
  \label{eq:exp_theta_X_outer_tensor}
  e^{-i \mangle X} \ket|z> = \bigotimes_{l = 1}^n \ab(\cos \mangle \ket|z_l> - i \sin \mangle \ket|z_l \oplus 1>).
\end{equation}
    We would now like to express this state as a superposition of computational basis states, which can be achieved by reconstructing each of its components $\ket|k \in \mathds{F}_2^n>$.  From \cref{eq:exp_theta_X_outer_tensor}, we see that the amplitude of $\ket|k>$ acquires a multiplicative pre-factor of either $-i\sin \mangle$ for each flipped, or $\cos \mangle$ for each preserved qubit, respectively. In other words, the eventual amplitude of state $\ket|k>$ is given by $f(\mangle, z, k) = \ab(\cos \mangle)^{n - d_H (z, k)} \cdot\ab(-i \sin \mangle)^{d_H (z, k)}$.
\end{proof}

Now that we have characterised the mixer and phase separation layers individually, we have the tools to further our analysis. Next in line is the optimisation landscape, which is a central point of interest when analysing \QAOA circuits.

\begin{figure*}
    \centering
    \begin{subfigure}{\textwidth}
    \includegraphics{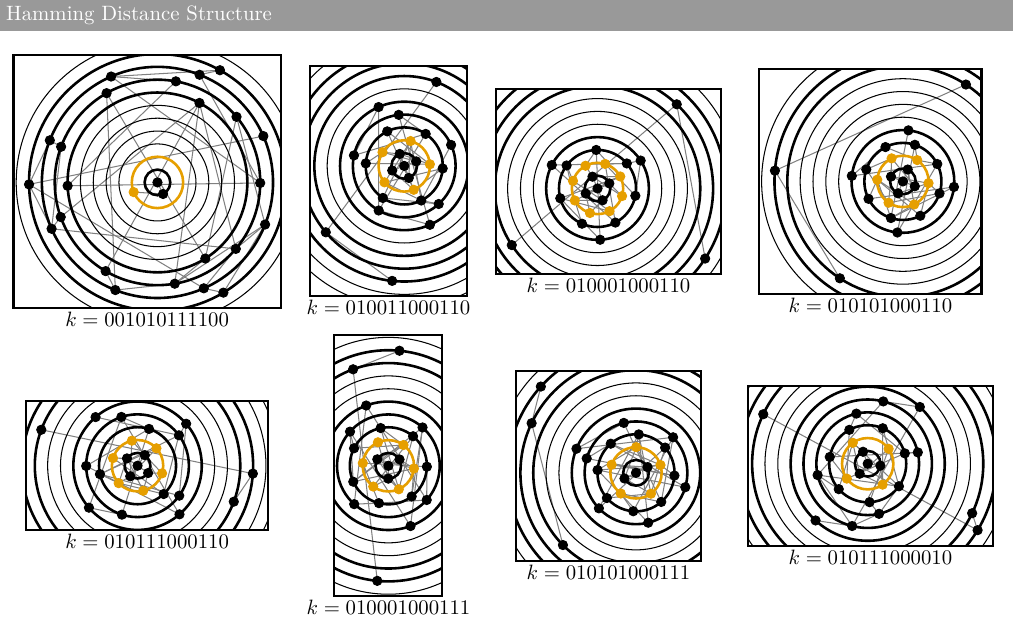}
    \caption{%
        \justifying
        \Cref{thm:optimisation_landscape} relates \QAOA ($p = 1$) landscapes with 
        the Hamming distance structures of instance solution spaces by counting
        equally distanced states in $\#_d \ab(k) = \ab|\ab\{ z \in T \mid d_H \ab(z, k) = d \}|$.
        We can visualize this by drawing the graph of all solutions $k$ in the 
        solution space $T$, where two solutions $k, z \in T$ are connected by
        an edge iff they have hamming distance $d_H (z, k) = 1$, \ie you can 
        reach one from the other by flipping just one bit. Now, to visualize 
        $\#_d \ab(k)$ we place $k$ in the center and order all $z \in T$ in 
        orbits around $k$, where $z$ is in the $i$-th orbit iff $d_H \ab(z, k) = i$.
        Then, $\#_d \ab(k)$ is the number of vertices on the $d$-th orbit.
    }
    \end{subfigure}
    \par\vspace{2em}
    \begin{subfigure}{\textwidth}
    \includegraphics{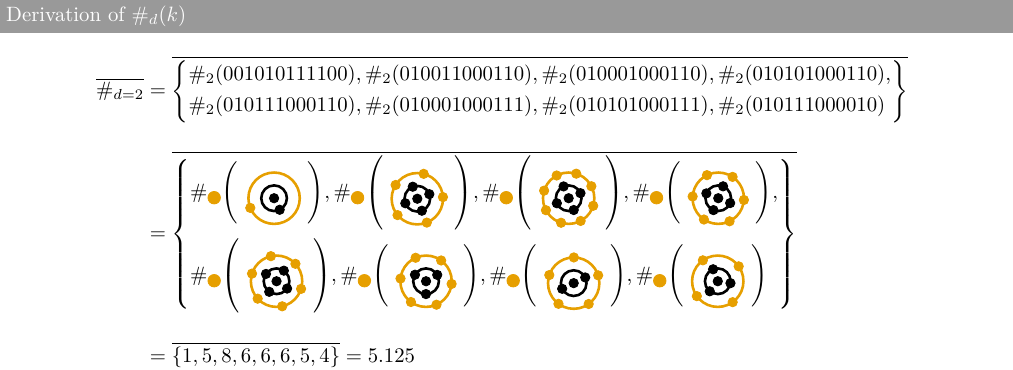}
    \caption{%
        \justifying
        By looking at the upper visualisation of solutions of a randomly generated
        \SAT formula (see \cref{sec:sat_example}), the global similarities between
        equal orbits across different instances become apparent. This is the key 
        motivation for \cref{thm:approx}, where global problem structures are captured 
        by averaging $\#_d \ab(k)$ over all $k \in T$. Making use of this visualisation, 
        we can get a graphical intuition on how to derive $\overline(\#_d)(k)$, 
        as depicted for $d = 2$ and the 8 solutions plotted above.
    }
    \end{subfigure}
    \caption{%
        \justifying
        A central part of our approximation theorem is the idea of grouping equal
        Hamming distances. How this captures structural information of a problem 
        solution space is visualized in (a). Mathematically this concept is expressed
        by a counting function $\#_d$, its derivation is depicted in (b).
    }
    \label{fig:hamming_dist_struct_vis}
\end{figure*}

Now that we have characterised the mixer and phase separation layers individually, we have the tools to further our analysis. Next in line is the optimisation landscape, which is a central point of interest when analysing \QAOA circuits. In \cref{thm:optimisation_landscape} we use that the landscape $F_1$ is a scaled sum of terms $c_k$ that for fixed angles $\beta$ and $\phi$ only depend on the Hamming distances between $k$ and all other states $z \in \mathds{F}_2^n$, which allows us to simplify the equation by collecting all terms with equal hamming distance. For a visualization see \cref{fig:hamming_dist_struct_vis}.

\begin{lemma}
  \label{thm:optimisation_landscape}
  A \QAOA circuit with $p=1$ constructed for general decision problems induces the optimisation landscape
  \begin{equation}
    \label{eq:F_1}
    F_1(\mangle, \pangle) = \frac{1}{2^n} \sum_{k \in T} \ab|c_k\ab(\mangle,\pangle)|^2
  \end{equation}
  with 
  \begin{equation}
    \label{eq:c_k}
      c_k\ab(\mangle,\pangle) \coloneqq \sum_{d=0}^n \ab(\#_d (k) \ab(e^{-i \pangle} - 1) + \binom{n}{d}) f_n \ab(\mangle, d)
  \end{equation}
    where $\#_d (k) = \ab|\ab\{z \in T \mid d_H\ab(z,k) = d \}|$ and $f_n \ab(\mangle, d) = \ab(\cos \mangle)^{n-d} (-i \sin \mangle)^d$. Here, $d_H\ab(z,k)$ is the Hamming distance between the bit strings of the binary representations of $z$ and $k$.
\end{lemma}

\begin{proof}
  By fixing $p = 1$ in \cref{eq:qaoa_state}, we obtain the state $\ket|\mangle, \pangle>_1 = e^{-i \mangle X} e^{-i \pangle C} \ket|+>^{\otimes n}$. After applying the phase separating gates to $\ket|+>^{\otimes n}$ and linearly pulling the mixer operators into the sums, we arrive at $\frac{1}{\sqrt{2^n}} \ab(e^{-i \pangle} \sum_{z \in T} e^{-i \mangle X} \ket|z> + \sum_{z \notin T} e^{-i \mangle X} \ket|z>)$. From \cref{thm:mixer_amplitudes_f}, it follows that
  \begin{equation}
  \begin{aligned}
    \label{eq:qaoa_state_1_mixer_applied}
    \ket|\mangle, \pangle>_1 = \frac{1}{\sqrt{2^n}} \Biggl(&e^{-i \pangle} \sum_{z \in T} \sum_{k \in \mathds{F}_2^n} f\ab(\mangle, z, k) \ket|k> + \\
                                                         &\sum_{z \notin T} \sum_{k \in \mathds{F}_2^n} f\ab(\mangle, z, k) \ket|k>\Biggr)  
  \end{aligned}
  \end{equation}
    By reordering the sums and factoring out $\ket{|k>}$ in \cref{eq:qaoa_state_1_mixer_applied}, we can express the state $\ket|\mangle, \pangle>_1$ as $\frac{1}{\sqrt{2^n}} \sum_{k \in \mathds{F}_2^n} c_k\ab(\mangle,\pangle) \ket|k>$ with $c_k\ab(\mangle,\pangle) = e^{-i \pangle} \sum_{z \in T} f\ab(\mangle, z, k) + \sum_{z \notin T} f\ab(\mangle, z, k)$
    (note that we omit parameters \(\ab(\mangle, \pangle)\) on \(c_{k}\) and other quantities below when the dependency is clear from the context). If we apply $C$ to $\ket|\mangle, \pangle>$ we basically filter out all $\ket|k> \notin T$ by linearly pulling $C$ into the sum, therefore obtaining $C \ket|\mangle, \pangle> = \frac{1}{\sqrt{2^n}} \sum_{k \in T} c_k\ab(\mangle,\pangle) \ket{|k>}$. We conclude that
  \begin{align}
    F_1(\mangle, \pangle) &=
    \braket<\mangle, \pangle|_1 C |\mangle, \pangle>_1\nonumber\\
    &= \frac{1}{2^n} \ab(\sum_{k \in \mathds{F}_2^n} \cconj{c_k\ab(\mangle,\pangle)} \bra<k|) \ab(\sum_{k \in T} c_k\ab(\mangle,\pangle) \ket|k>)\label{eq:F1_step2}\\ 
      &= \frac{1}{2^n} \sum_{k \in T} \cconj{c_k\ab(\mangle,\pangle)} c_k\ab(\mangle,\pangle) \\
      &= \frac{1}{2^n} \sum_{k \in T} \ab|c_k\ab(\mangle,\pangle)|^2.\label{eq:F1_step3}
  \end{align}
    The step from \cref{eq:F1_step2} to \cref{eq:F1_step3} follows since all pairs in $\ab\{\ket|k> \mid k \in \mathds{F}_2^n\} \times \ab\{\ket|k> \mid k \in T\}$ comprise orthogonal states. 
    Note that $c_k\ab(\mangle,\pangle)$ essentially contains a sum over all $z \in \mathds{F}_2^n$ with an optional phase factor of $e^{-i \pangle}$ if $z \in T$. Thus, we find that
    \begin{equation}
        \label{eq:ck_over_all_states}
        \begin{aligned}
            c_k\ab(\mangle,\pangle) = \sum_{z \in \mathds{F}_2^n} &\ind_T(z) e^{-i \pangle} f\ab(\mangle, z, k) + \\
            &\ab(1 - \ind_T(z)) f\ab(\mangle, z, k).
        \end{aligned}
    \end{equation}
    Here $\ind_T(z) = 1$ iff $z \in T$ and otherwise $\ind_T(z) = 0$. By recalling the definition of $f\ab(\mangle, z, k)$ in \cref{eq:f_with_states}, we recognise $\ab(\cos \mangle)^{n - d} \ab(-i \sin \mangle)^d$ as the basic shape of $f$, where $d$ is the Hamming distance between two concrete states $k,z \in \mathds{F}_2^n$. This means that for a fixed angle $\mangle_{\text{c}}$, $f(\mangle_{\text{c}}, k, z)$ is actually a function of the Hamming distance between two states. Therefore, while the sum in \cref{eq:ck_over_all_states} iterates over \emph{all} $2^n$ states $z$, there are effectively only $n+1$ \emph{different} basic terms in the sum---one for each possible distance $d$. Thus, if we count the number of occurrences $\#_d (k) = \ab|\ab\{z \in T \mid d_H\ab(z,k) = d \}|$ of each distance $d$ between $k$ and other states in the target set, we can reorder the sum while still capturing all possible bit-flip mappings from $z$ to $k$:
  \begin{align*}
      c_k\ab(\mangle,\pangle) = \sum_{d = 0}^n &\#_d (k) e^{-i \pangle} f_n \ab(\mangle, d) + \\
      &\ab(\binom{n}{d} - \#_d (k)) f_n \ab(\mangle, d)
  \end{align*}
  with 
  \begin{equation}
    f_n \ab(\mangle, d) \coloneqq \ab(\cos \mangle)^{n - d} \ab(-i \sin \mangle)^d.
  \end{equation}
    Now it is simply a matter of factoring out $f_n$ and $\#_d$ to arrive at \cref{eq:c_k}.
\end{proof}

\section{The \QAOA Approximation Theorem}\label{sec:approx_thm}
We now state and prove our main result.

\subsection{Intuition}
\begin{figure*}[htbp]
    \includegraphics[trim={0 3mm 0 10mm}, clip, width=\linewidth]{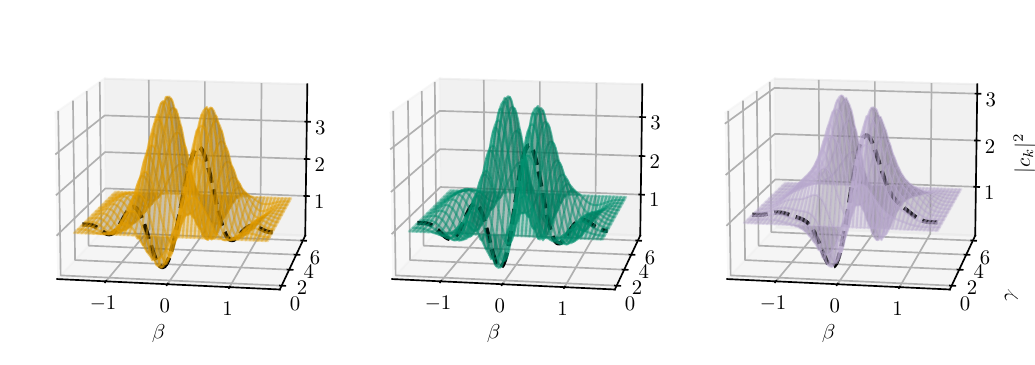}
    \caption{Three landscape components $\ab|c_k\ab(\mangle,\pangle)|^2$ for $k \in \ab\{10, 13, 14\} = T$ and a state space of dimension $n = 5$ are depicted in this order from left to right. The landscapes are represented by an array of vertical cross sections along $\mangle$.}
    \label{fig:3d_overview}
\end{figure*}
\begin{figure}[htbp]
    \includegraphics{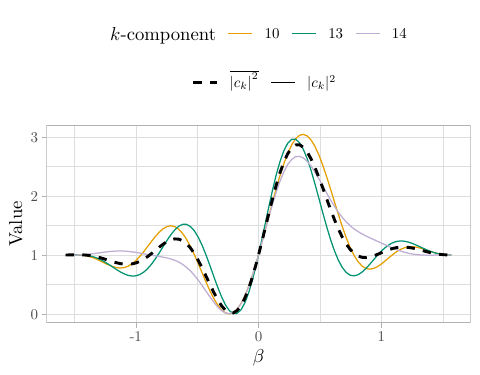}
    \caption{Cross section of the optimisation landscape components $\ab|c_k\ab(\mangle,\pangle_{\text{c}})|^2$  depicted in \cref{fig:3d_overview}
    at \(\pangle_{\text{c}}=1.2\) for $k \in \ab\{10, 13, 14\} \subset T$ (note matching colours). The dashed line shows \(E\ab(|c_k\ab(\mangle,\pangle_{\text{c}})|^2)\) for these three components, and illustrates how the mean value captures the globally relevant features of instance-specific information.}
    \label{fig:c_k_components}
\end{figure}

A straightforward description of an optimisation landscape might be directly derived by considering all point to point interactions between all states in superposition. Obviously, there are exponentially many such state interactions to consider. In \cref{sec:opt_landscape}, we captured those effects in $f\ab(\mangle, k, z)$ (see \cref{eq:f_with_states}) and showed that there are only $n + 1$ different outcomes that can occur based on the Hamming distance between two states. Now we use this reduction in complexity of the effect domain to derive an approximation theorem for the \emph{expected optimisation landscape} $E\ab(F_1)$ of a specific problem.

\subsection{Formalisation}
In \cref{sec:opt_landscape}, we presented a closed form of the optimisation landscape based on the Hamming distances between states. 
The components $\ab|c_k\ab(\mangle,\pangle)|^2$ in \cref{eq:c_k} play a crucial role
in $F_1\ab(\mangle, \pangle)$ (see \cref{eq:F_1}). As they are functions of the optimisation angles, we find it prudent to first obtain an intuitive visual understanding of their 
interrelationship. Each contribution $\ab|c_k\ab(\mangle,\pangle)|^2$ can be seen as a landscape for a specific value of \(k\) itself. Consider \cref{fig:3d_overview} that depicts, as an illustration, vertical slices of three $\ab|c_k\ab(\mangle,\pangle)|^2$ for $k \in T = \ab\{10, 13, 14\}$ and dimension $n = 5$. The macroscopic similarities between the three are visually obvious, which motivate the idea to take the mean over all $\ab|c_k\ab(\mangle,\pangle)|^2$. To further quantify this idea, we look at the overlaid cross sections for all three components compared with their mean in \cref{fig:c_k_components}. Here, we evaluate $\ab|c_k\ab(\mangle,\pangle)|^2$ for all $k \in T = \ab\{10, 13, 14\}$ across $0 \leq \mangle \leq \pi$ at a fixed $\pangle_{\text{c}} = 1.2$ and calculate the mean of all $\ab|c_k\ab(\mangle,\pangle_{\text{c}})|^2$ at each point. The chosen constant 1.2 does not have any particular computational or physical significance, but results in a \enquote{typical} two-dimensional sub-space of the optimisation landscape. We again observe how the mean value nicely captures the higher level structures shared between all individual components. In fact, this intuition can be expressed as a precise mathematical relationship, as $F_1$ is nothing more than this mean value of $c_k$ terms scaled by the relative cardinality $\frac{\ab|T|}{2^n}$ of the target space $T$, since
\begin{equation}
    \label{eq:F_as_average}
    \begin{aligned}
        \frac{\ab|T|}{2^n} \overline{\ab|c_k\ab(\mangle,\pangle)|^2} &= \frac{\ab|T|}{2^n} \frac{1}{\ab|T|} \sum_{k \in T} \ab|c_k|^2 \\
        &=  \frac{1}{2^n} \sum_{k \in T} \ab|c_k|^2 = F_1\ab(\mangle,\pangle).
    \end{aligned}
\end{equation}
As remarked above, the global structure of $F_1$ is very similar across different instances of a given problem. This renders the \emph{expected value} of an optimisation landscape for a problem an interesting object of study, as it captures the salient properties \emph{across} instances. In particular, we argue 
that the ability to efficiently approximate  
is well suited to improve the understanding of \QAOA, as we show in \cref{sec:application},
and also leads to practical implications that can help
optimise the use of \QAOA, as we detail in \cref{sec:utility}.

Before that, and based on these observations, let us however first state and prove our main theorem for efficiently approximating the expected value of $F_1 \ab(\mangle, \pangle)$ across instances of a computational problem. Recall that the specific optimisation Landscape of a single instance can be expressed as the mean
of its $\ab|c_k \ab(\mangle, \pangle)|^2$ weighted by the relative cardinality of its solution space compared to the full state space ($|T| 2^{-n}$). Now we show, that we can approximate the expected landscape of a random problem instance by calculating the expected values of these two parameters: $E\ab(\overline{\ab|c_k \ab(\mangle, \pangle)|^2})$ and $E\ab(|T|)$.

\begin{theorem}
  \label{thm:approx}
     Let $\#_d (k)$ be defined as in \cref{thm:optimisation_landscape}, then we can approximate the expected value of $F_1 (\mangle, \pangle)$ for a random instance of a problem by
  \begin{displaymath}
      \tilde{E}\ab(F_1 (\mangle, \pangle)) = \frac{E\ab(\ab|T|)}{2^n} E\ab(\overline{\ab|c_k\ab(\mangle,\pangle)|^2})
  \end{displaymath}
  with
  \begin{equation}
      \label{eq:E_ck_mean}
      E\ab(\overline{\ab|c_k\ab(\mangle,\pangle)|^2}) = \sum_{d_1, d_2 = 0}^n w_{d_1, d_2}\ab(\pangle) f_n \ab(\mangle, d_1) f_n \ab(\mangle, d_2)^* ,
  \end{equation}
    where $w_{d_1, d_2} \ab(\pangle)$ is
  \begin{displaymath}
  \begin{aligned}
      w_{d_1, d_2}\ab(\pangle) = &E\ab(\overline{\#_{d_1}\ab(k) \#_{d_2}\ab(k)}) \ab(e^{-i\pangle} - 1)\ab(e^{i\pangle} - 1) + \\
                                 &E\ab(\overline{\#_{d_1}\ab(k)}) \ab(e^{-i \pangle} - 1) \binom{n}{d_2} + \\
                                 &E\ab(\overline{\#_{d_2}\ab(k)}) \ab(e^{i \pangle} - 1) \binom{n}{d_1} + 
                                 \binom{n}{d_1} \binom{n}{d_2}.
  \end{aligned}
  \end{displaymath}
    We calculate mean values over all $k$ in the instance specific target set $T$. The absolute approximation error is bounded by $\ab|E\ab(F_1 \ab(\mangle, \pangle)) - \tilde{E}\ab(F_1 \ab(\mangle, \pangle))| \leq \sqrt{\Var\ab(\ab|T|) \Var\ab(\overline{\ab|c_k|^2})}$.
\end{theorem}

\begin{proof}
    We start by reformulating $|c_k \ab(\mangle, \pangle)|^2$. Recall that $|c_k \ab(\mangle, \pangle)|^2 = c_k \ab(\mangle, \pangle) c_k \ab(\mangle, \pangle)^*$ and $c_k \ab(\mangle, \pangle) = \sum_{d = 0}^n \ab(\#_d (k) \ab(e^{-i \pangle} - 1) + \binom{n}{d}) f_n \ab(\mangle, d)$. Then we use the distributivity of the complex conjugate over addition and multiplication to pull it into the sum:
    \begin{displaymath}
    \begin{aligned}
        c_k \ab(\mangle, \pangle)^* &= \sum_{d = 0}^n \ab(\#_d (k) \ab(\ab(e^{-i \pangle})^* - 1) + \binom{n}{d}) f_n \ab(\mangle, d)^* \\
                                    &= \sum_{d = 0}^n \ab(\#_d (k) \ab(e^{i \pangle} - 1) + \binom{n}{d}) f_n \ab(\mangle, d)^*
    \end{aligned}
    \end{displaymath}
    where the second equality follows from $(e^z)^* = e^{z^*}$. From this we see, that
    \begin{alignat*}{2}
    \ab|c_k \ab(\mangle, \pangle)|^2 = &\Bigg(\sum_{d_1 = 0}^n \bigg(&&\#_{d_1} (k) \ab(e^{-i \pangle} - 1) + \\
                                                         & &&\binom{n}{d_1}\bigg) f_n \ab(\mangle, d_1)\Bigg) \cdot \\
                                           &\Bigg(\sum_{d_2 = 0}^n (&&\#_{d_2} (k) \ab(e^{i \pangle} - 1) + \\
                                                         & &&\binom{n}{d_2}) f_n \ab(\mangle, d_2)^*\Bigg) 
    \end{alignat*}
    and further
    \begin{align*}
        &\ab|c_k \ab(\mangle, \pangle)|^2 = \sum_{d_1, d_2 = 0}^n \Bigg( \binom{n}{d_1} \binom{n}{d_2} + \\
        &\#_{d_1} (k) \ab(e^{-i \pangle} - 1) \binom{n}{d_2} +  \#_{d_2} (k) \ab(e^{i \pangle} - 1) \binom{n}{d_1} + \\
        &\#_{d_1} (k) \#_{d_2} (k) \ab(e^{-i \pangle} - 1) \ab(e^{i \pangle} - 1) \Bigg) f_n \ab(\mangle, d_1) f_n \ab(\mangle, d_2)^*
    \end{align*}
    Now from $\overline{\ab|c_k \ab(\mangle, \pangle)|^2} = \frac{1}{|T|} \sum_{k \in T} \ab|c_k \ab(\mangle, \pangle)|^2$ it follows by reordering the sum that
    \begin{align*}
        &\overline{\ab|c_k \ab(\mangle, \pangle)|^2} = \sum_{d_1, d_2 = 0}^n \Bigg(\binom{n}{d_1} \binom{n}{d_2} + \\
        &\overline{\#_{d_1} (k)} \ab(e^{-i \pangle} - 1) \binom{n}{d_2} + \overline{\#_{d_2} (k)} \ab(e^{i \pangle} - 1) \binom{n}{d_1} + \\
        &\overline{\#_{d_1} (k) \#_{d_2} (k)} \ab(e^{-i \pangle} - 1) \ab(e^{i \pangle} - 1) \Bigg)  f_n \ab(\mangle, d_1) f_n \ab(\mangle, d_2)^*
    \end{align*}
    where $\overline{\#_d (k)} = \frac{1}{\ab|T|} \sum_{k \in T} \#_d (k)$. Finally, \cref{eq:E_ck_mean} follows from the linearity of the expected value. From \cref{eq:F_as_average} we deduce that $E\ab(F_1) = E \ab(\frac{\ab|T|}{2^n} \overline{\ab|c_k|^2}) = \frac{E\ab(\ab|T|)}{2^n} E\ab(\overline{\ab|c_k|^2}) + \Cov\ab(\ab|T|, \overline{\ab|c_k|^2})$. Therefore we conclude that $\ab|\tilde{E}\ab(F_1) - E\ab(F_1)| = \left|\Cov\ab(\ab|T|, \overline{\ab|c_k|^2})\right| \leq \sqrt{\Var\ab(\ab|T|) \Var\ab(\overline{\ab|c_k|^2})}$.
\end{proof}

\subsection{Generality}

The approximation theorem only operates on the Hamming distance structure of problem target spaces. Recall, however, that we arrive at the target space $T$ by defining it through an isomorphism from the set of all solutions $\{z \in \mathds{F}_2^n \mid c(z) = 1\}$ to a subset of quantum basis states $T = \{\ket|z> \mid c(z) = 1\}$, let that be $f: z \mapsto \ket|z>$. A more complex constraint Hamiltonian might also need to operate on ancilla qubits that are not part of the target space, possibly embedding $f$ in a structure like $g: xz \mapsto \ket|h(x)> \otimes \ket|z>$. This additional degree of freedom allows for inconsistent side effects on anchilla registers while still maintaining functional correctness. As a result, this could thus alter the hamming distance between pairs of states encoding the same solutions. As a consequence the approximation theorem would require a mechanism to take computations on arbitrary anchilla registers into account. We would thus lose the elegance of being independent of concrete Hamiltonian constructions. We now show that for every problem in \NP, it is possible to efficiently construct an ancilla register independent constraint Hamiltonian. This guarantees a wide applicability of \cref{thm:approx} only focusing on the concrete solution structure of the problem at hand. For the following theorem and its proof we use verifier definition of \NP, that is: A problem is in \NP if and only if there exists an algorithm $v: \mathds{F}_2^n \to \mathds{F}_2$ that efficiently verifies solutions $z \in \mathds{F}_2$, \ie $v(z) = 1$ iff $z$ is a valid solution to the problem and otherwise $v(z) = 0$.

\begin{theorem}
    \label{thm:ancilla_invariance}
    For every problem in \NP there exists a family of constraint two-level Hamiltonians $\ab\{C_n\}_{n \in \mathds{N}}$ that can be efficiently constructed such that 
    \begin{equation}
        C_n\ket|z>\ket|\bm{0}> = \begin{cases}
            \ket|z>\ket|\bm{0}> &\quad z \:\:\text{proves the decision property} \\
            0                       &\quad \text{otherwise}
        \end{cases}
    \end{equation}
    for proofs $z$ of length $n$.
\end{theorem}

\begin{proof}
    By definition a problem is in \NP if and only if there exist an efficient verifier $v$ that returns $v(z) = 1$ on all valid proofs $z$ of the decision property and $v(z) = 0$ on all other inputs. Thus, there also exits an efficiently constructable family of quantum circuits $\ab\{V_n\}_{n \in \mathds{N}} : V_n \ket|z>\ket|\bm{0}> = \ket|z>\ket|\bm{v}>$ with $\bm{v} = \ab(v(z), v_2, \dots, v_{p(n)})$, where $\ket|\bm{0}> = \ket|0>^{\otimes p(n)}$ is an ancilla register of size $p(n)$ for some polynomial $p$. Now, the construction of $C_n$ looks as follows: $C_n \coloneqq V_n^\dagger (\mathds{1}^{\otimes n} \otimes \ketbra|1><1| \otimes \mathds{1}^{p(n) - 1}) V_n$.
\end{proof}

\begin{remark}
    \Cref{thm:ancilla_invariance} ensures us that if a problem is in \NP, we only have to consider its target space $T$ of problems when applying \cref{thm:approx}. Remember that $T$ is isomorphic to the solution or proof space of a problem. In essence, \cref{thm:approx} addresses structural properties of the expected solution space of a problem. So, the combination of \cref{thm:approx} and \cref{thm:ancilla_invariance} shows the existence of macroscopic structures on a problem level that significantly influence the optimisation landscapes induced by \QAOA circuits for at least all \NP problems. 
\end{remark}

On a more practical note, this result means that we dont have to take computational side effects of concrete circuit implementations into account, when using \cref{thm:approx} to analyse the effect of problem structures on QAOA. Let's say for example, one wants to solve a \kclique problem, wich is the corresponding decision problem to the \maxclique optimisation problem. Then we can construct a QAOA circuit using \cref{thm:ancilla_invariance} to construct the constraint Hamiltonian needed for the problem layer. The resulting QAOA circuit will then solve the \kclique problem by preparing a superposition $\sum_i \alpha_i \ket|z_i>\ket|\bm{0}>$ of states where $z_i$ representing cliques of size $k$. For a detailed description of exactly this example see \cref{sec:kclique}.

\subsubsection{A Note about \NP}
We now want to argue why covering \NP is practically sufficient to prove generality for our framework. We acknowledge that in the quantum computing community usually optimization problems like for example the ever present \maxcut problem are the subject of investigation. Thus, it might seem unintuitive to base a framework on a decision problem class like \NP. In fact, most problems relevant in practice are in some variant part of \NP, see the seminal list of \NP problems \cite{Karp1972} by \citeauthor{Karp1972}. Indeed, foundational work by \citeauthor{lucas2014ising}  that arguably significantly contributed to the commercialisation of quantum computing \cite{lucas2014ising} provided Ising formulations of many important \NP problems, bridging the gap between classical industrial problems solving usually powered by highly sophisticated SAT-solvers \cite{heule2024proceedings}, which are basically \NP solvers, and quantum computing at that time. 

Coming back now to the other perspective, starting with a optimisation problem, how do we apply our framework. Let's consider the \maxcut problem for instance. Instead of asking for an optimal cut, we reformulate the problem to ask whether there exists a cut above a certain threshold $k \in \mathds{N}$. We can as proven above for all optimization problems which have a threshold variant in \NP. In practice this applies many interesting problems like \maxcut, Partition, Vertex Cover, etc. \cite{Karp1972}. Then we formulate a Hamiltonian that \emph{verifies} states encoding a cut satisfying the threshold and apply our framework to it. Compared to the classical SAT-solver case, the Hamiltonian plays the roll of the stat formula and our QAOA approach takes the place of the SAT-solver.

\section{Application}
\label{sec:application}
After having set the foundations for a methodology to understand
structural properties of optimisation problems across instances, let us now
commence with applying the framework to several concrete examples. We consider five different subject problems respectively scenarios (uniform random sampling, clustered 
sampling,  Boolean satisfiability, k-clique, and one-way functions in the form of qr factoring) that build on each other to best introduce the application of the \QAOA approximation theorem on real problems. Overall, the selection of examples is carefully curated to highlight different aspects of using our approximation theorem: We first compare the two possible approaches of either analytically or empirically analysing the target space structure of a problem at hand. Then we demonstrate with a purposefully constructed sampling method that the existence of a stochastic dependence of two states with a certain Hamming distance in a random instance significantly influences he optimisation landscape. Following that, we use \SAT as a first straight forward real world example with an easy to construct constraint Hamiltonian. After that, we show that for all problems in \NP, even with more complicated constraints, there is a construction for an appropriate constraint Hamiltonian that satisfies the preconditions of our approximation theorem. As a final example, we highlight the case of integer factoring to argue that problems based on
one-way functions are interesting subjects for our methods as their target space can be relatively easy characterised.

\subsection{Uniform Random Sampling}
\label{sec:uniform_sampled}
Central in our theory is the target space $T$ of the constraint projector $C$. Given a concrete interpretation, every state in $T$ corresponds to a problem solution.\footnote{Note that given different interpretations, a state in $T$ can encode different solutions of different problems, too.} Thus, our notion of target spaces $T$ is an abstraction of concrete problem solution spaces. Sampling a random target space $T$ equals sampling a random problem minus the abstraction of a target state interpretation. Furthermore, the description of a sampling procedure for $T$ defines an abstract random problem, where the problem itself becomes a random variable in the probabilistic point of view. To provide a smooth onramp to more complex examples further below, we start by simply sampling $T$ by randomly drawing states from the state space with uniform probability.  

From \cref{thm:approx}, we see that $E\ab(F_1)$ can be efficiently approximated if the quantities (a) $\ab|T|$, (b) $E\ab(\overline{\#_d (k)})$ for all $0 \leq d \leq n$, and (c) $\Cov \ab(\overline{\#_{d_1} (k)}, \overline{\#_{d_2} (k)})$ for all $0 \leq d_1, d_2 \leq n$ are provided. Recall that the expected value is calculated over all instances of a specific problem, thus $\overline{\#_d (k)}$ are random variables describing the mean value $\frac{1}{\ab|T|} \sum_{k \in T} \#_d(k)$ over the target set $T$ of a random problem instance. There are, in general, two approaches: 
We can, if possible, determine the distribution of the random variables $\overline{\#_d (k)}$
by analytical means, or by an empirical
numerical approach. For uniform sampling as considered in this
motivating example, it is relatively straightforward to execute
the necessary analytic calculations. 

\subsubsection{Analytic Approach}
\begin{definition}
    We consider an urn model with balls of $m \in \mathds{N}$ different colours. The population of all balls in the urn is described by $\bm{X} = \ab(X_1, X_2, \dots, X_m)$, where $X_i$ is the number of balls with colour $i$. Then, if $n$ balls are drawn without replacement, the probability of having $x_i$ balls of colour $i$ in the sample is given by the 
    multivariate hypergeometric distribution
    \(P\ab(\bm{x}) = \hypgeom \ab(\bm{X}, \bm{x})\)
    with $\bm{x} = \ab(x_1, x_2, \dots, x_m)$. Recall that, by textbook knowledge, the hypergeometric distribution is defined by 
    \(P(\bm{x}) = \prod_{i=1}^m \binom{X_i}{x_i}/\binom{N}{n}\),
         with \(N \coloneqq \sum_{i=1}^m X_i\).
\end{definition}

    The size of $T$ is trivially given, as we always sample a fixed number of states. In the following discussion we will encounter expected values over different sample spaces, we will differentiate this by $E_T(\cdot)$ and $E_{\mathcal{I}}(\cdot)$ being defined to be the mean over all states in the target space $T$ and over all instances in the set of problem instances $\mathcal{I}$. No subscript is used if the sample space is clear from context. Further note that $\overline{\#_d(k)}$ is the sample mean of $T$ where $T$ is sampled from the sample space of all problem instances $\mathcal{I}$. Therefore, $E_{\mathcal{I}}\ab(\overline{\#_d (k)}) = E_T \ab(\#_d (k))$ as the sample mean $\overline{\#_d(k)}$ is an unbiased estimator of the expected value $E_T \ab(\#_d (k))$. The same also holds for $\overline{\#_{d_1}(k) \#_{d_1}(k)}$. To calculate $E_T\ab(\#_d (k))$, we need to know the distributions of $\ab\{\#_d (k)\}_{d=0}^n$. To compute $E_{\mathcal{I}}\ab(\overline{\#_{d_1} (k) \#_{d_2} (k)}) = E_T \ab(\#_{d_1} (k) \#_{d_2} (k))$ we use that $E_T \ab(\#_{d_1}(k) \#_{d_2} (k)) = E_T \ab(\#_{d_1}(k)) E_T \ab(\#_{d_2}(k)) + \Cov\ab(\#_{d_1} (k), \#_{d_2} (k))$. For $\Cov \ab(\#_{d_1} (k), \#_{d_2} (k))$, the joint probability distribution of $\ab\{\ab(\#_{d_1} (k), \#_{d_2} (k))\}_{d_1, d_2 = 0}^n$ is required.

As $T$ is uniformly sampled from the complete state space, this basically leads to an urn model without replacement. Therefore, except for some edge cases, the random variables $\#_d (k)$ can be described by a Hypergeometric distribution, see \cref{thm:Pd}.

\pagebreak
\begin{widetext}
\begin{lemma}
    \label{thm:Pd}
    In case of uniform target sampling the probabilities $P\ab(\#_d (k) = x)$ and $P\ab(\#_{d_1} (k) = x_1, \#_{d_2} (k) = x_2)$ are defined as follows:
    \begin{equation}
        \label{eq:Pd}
        P\ab(\#_d (k) = x) = 
        \begin{cases}
            1 &\quad d = 0 \wedge x = 1 \\
            0 &\quad d = 0 \wedge x \neq 1 \\
            \hypgeom \ab(\bm{X}, \bm{x}) &\quad \text{otherwise}
        \end{cases}
    \end{equation}
    with $\bm{X} = \ab(\binom{n}{d}, 2^n - 1 - \binom{n}{d})$ and $\bm{x} = \ab(x, \ab|T| - 1 - x)$. Furthermore, 
    \begin{equation}
        \label{eq:Pd1d2}
        P\ab(\#_{d_1} (k) = x_1, \#_{d_2} (k) = x_2) = \begin{cases}
            P\ab(\#_{d_1} (k) = x_1) P\ab(\#_{d_2} (k) = x_2) &\quad d_1 = 0 \vee d_2 = 0 \\
            P\ab(\#_{d_1} (k) = x_1) &\quad d_1 = d_2 \wedge x_1 = x_2 \\
            0 &\quad d_1 = d_2 \wedge x_1 \neq x_2 \\
            \hypgeom \ab(\bm{X}, \bm{x}) &\quad \text{otherwise}
        \end{cases}
    \end{equation}
    with $\bm{X} = \ab(\binom{n}{d_1}, \binom{n}{d_2}, 2^n - 1 - \binom{n}{d_1} - \binom{n}{d_2})$ and $\bm{x} = \ab(x_1, x_1, |T| - 1 - x_1 - x_2)$.
\end{lemma}
\end{widetext}

\begin{proof}
    This proof is structured in two parts. We start by showing the correctness of \cref{eq:Pd}, which in turn necessitates distinguishing between two cases. As $\#_d(k)$ describes the distance relationship of a random state in $T$ to all other states in $T$ (including itself), $\#_0(k)$ is always 1, from which the first two cases of \cref{eq:Pd} follow. For $d \geq 0$, we need to 
    consider the sampling process of $T$. Assume we start with just one random state in $T$, and then sample the rest. This is described by an urn model with $X_1 = \binom{n}{d}$ balls of colour $d$, and $X_2 = 2^n - 1 - \binom{n}{d}$ other balls. Note that by picking a random reference state in $T$, there are only $2^n - 1$ balls left in the urn to draw the $\ab|T| - 1$ states left to fill the target space. Therefore, the probability of sampling $x$ states of Hamming distance $d$ from the random initial state is given by
    \begin{displaymath}
    P\ab(\#_d (k) = x) = \hypgeom\bigl(\ab(X_1, X_2), 
                                    \ab(k, \ab|T| - 1 - x)\bigr),
    \end{displaymath}
    which concludes the proof of \cref{eq:Pd}.

    We need to consider three cases for \cref{eq:Pd1d2}. If $d_1 = 0$, the only possible samples $(x_1, x_2)$ must have $x_1 = 1$, with the same reasoning as above. Thus, $P\ab(\#_{d_1} (k) = x_1, \#_{d_2} (k) = x_2) = \delta_{1, x_1} P\ab(\#_{d_2} (k) = x_2)$, where $\delta_{i,j}$ is the Kronecker delta function. Further note that $\delta_{1, x_1} = P\ab(\#_{d_1} (k) = x_1)$. The case of a vanishing second argument, $d_2 = 0$, is symmetric to $d_1 = 0$. Therefore, if follows that $P\ab(\#_{d_1} (k) = x_1, \#_{d_2} (k) = x_2) = P\ab(\#_{d_1} (k) = x_1)P\ab(\#_{d_2} (k) = x_2)$ iff $d_1 = 0$ or $d_2 = 0$. If $d_1 = d_2$,
    then obviously $x_1 = x_2$ needs to be satisfied, in which case $P\ab(\#_{d_1} (k) = x_1, \#_{d_2} (k) = x_2)$ reduces to $P\ab(\#_{d_1} (k) = x_1) = P\ab(\#_{d_2} (k) = x_2)$. For all other $d_1$ and $d_2$, re-consider the sampling process from an urn without replacement: This time, we are interested in two colours $d_1$ and $d_2$. Consequently, the urn contains $\binom{n}{d_1}$ balls of colour $d_1$, $\binom{n}{d_2}$ balls of colour $d_2$, and $2^n - 1 - \binom{n}{d_1} - \binom{n}{d_2}$ other balls. The probability of obtaining $k_1$ states of Hamming distance $d_1$ and $k_2$ states of Hamming distance $d_2$ in a sample of size $\ab|T|$ is therefore given by
    \begin{displaymath}
        P\ab(\#_{d_1} (k) = x_1, \#_{d_2} (k) = x_2) = \hypgeom\ab(\bm{X}, \bm{x})
\end{displaymath}
    with $\bm{X} = \ab(\binom{n}{d_1},\binom{n}{d_2}, 2^n - 1 - \binom{n}{d_1} - \binom{n}{d_2})$ and $\bm{x} = \ab(x_1, x_2, \ab|T| - 1 - x_1 - x_2)$.
\end{proof}

Now that we know the distribution of our random variables $\#_d (k)$, we can deduce properties like expected values and covariances. In \cref{thm:Ed}, we use \cref{thm:Pd} to determine $E_T \ab(\#_d (k))$ and $\Cov\ab(\#_{d_1} (k), \#_{d_2} (k))$.

\begin{lemma}
\label{thm:Ed}
\begin{equation}
    \label{eq:Ed}
    E_T\ab(\#_d (k)) = \begin{cases}
        1 &\quad d = 0 \\
        \ab|T| \frac{\binom{n}{d}}{2^n} &\quad \text{otherwise}
    \end{cases}
\end{equation}

\begin{equation}
    \label{eq:Covd1d2}
    \begin{aligned}
    &\Cov\ab(\#_{d_1} (k), \#_{d_2} (k)) = \\ 
    &\begin{cases}
        0 &\quad d_1 = 0 \vee d_2 = 0 \\
        \ab|T| \frac{2^n - \ab|T|}{2^n - 1} \frac{\binom{n}{d}}{2^n} \ab(1 - \frac{\binom{n}{d}}{2^n}) &\quad d_1 = d_2 = d > 0\\
        -\ab|T| \frac{2^n - \ab|T|}{2^n - 1} \frac{\binom{n}{d_1} \binom{n}{d_2}}{2^{2n}} &\quad \text{otherwise}
    \end{cases}
    \end{aligned}
\end{equation}
\end{lemma}

\begin{proof}
    In large, this proof follows the structure of the proof of \cref{thm:Pd}. We start by show the correctness of \cref{eq:Ed}. If $d = 0$, we trivially have $E\ab(\#_0) = 0 P\ab(\#_0 = 0) + 1 P\ab(\#_0 = 1) = 1$. For all $d > 0$, $P\ab(\#_d (k) = x)$ is described by the Hypergeometric distribution as in the second case of \cref{eq:Pd}, with an expected value of $\ab|T| \binom{n}{d} 2^{-n}$ as 
    required.

    We now prove \cref{eq:Covd1d2}. Assume that $d_1 = 0$, then with $\#_{d_1 d_2}(k) \coloneqq \#_{d_1} (k)\#_{d_2} (k)$ we have
    \begin{alignat*}{2}
        E\ab(\#_{d_1 d_2} (k)) &= &&\sum_k 0 \cdot \#_{d_2} (k) P\ab(\#_{d_1} (k) = 0, \#_{d_2} (k) = x) + \\ 
                                       &  &&\sum_k 1 \cdot \#_{d_2} (k) P\ab(\#_{d_1} (k) = 1, \#_{d_2} (k) = x) \\
                                       &= &&\sum_k \#_{d_2} (k) P\ab(\#_{d_1} (k) = 1, \#_{d_2} (k) = x)
    \end{alignat*}
    and of course $P\ab(\#_{d_1} (k) = 1) = 1$ for $d_1 = 0$ and thus $E\ab(\#_{d_1} (k) \#_{d_2} (k)) = \sum_k \#_{d_2} (k) P\ab(\#_{d_2} (k) = k) = E\ab(\#_{d_2} (k))$. Together with \cref{eq:Ed}, it follows that
    \begin{displaymath}
    \begin{aligned}
        \Cov\ab(\#_{d_1} (k), \#_{d_2} (k)) = &E\ab(\#_{d_1} (k) \#_{d_2} (k)) - \\
                                              &E\ab(\#_{d_1} (k)) E\ab(\#_{d_2} (k)) \\
                                    = &E\ab(\#_{d_2} (k)) - E\ab(\#_{d_2} (k)) = 0.
    \end{aligned}
    \end{displaymath}
    The case for $d_2 = 0$ is symmetric. We conclude that $\Cov\ab(\#_{d_1} (k), \#_{d_2} (k)) = 0$ iff $d_1 = 0$ or $d_2 = 0$. Note that $\Cov\ab(\#_d (k), \#_d (k)) = \Var(\#_d (k))$, which is given by
    \begin{displaymath}
        \Var(\#_d (k)) = \ab|T| \frac{2^n - \ab|T|}{2^n - 1} \frac{\binom{n}{d}}{2^n} \ab(1 - \frac{\binom{n}{d}}{2^n})
    \end{displaymath}
    since the sampling process is described by a multivariate hypergeometric distribution as shown in \cref{thm:Pd}. The third and second case also directly follow from the covariance of the multivariate hypergeometric distribution from \cref{thm:Pd}.
\end{proof}

With the two lemmas above, we can efficiently calculate $\tilde{E}\ab(F_1)$. This demonstrates a use case that benefits from our approximation theorem: Given a problem for which the target space structure with respect to Hamming distances between targets can be modelled analytically, the expected optimisation landscape can be approximately derived solely from this model. While instance specific analytic formulations of $F_1$ are known for the Ising model problem \cite{Ozaeta:2022} from a physics-centric
point of view, our approach accesses the problem from a previously unexplored angle 
that benefits from insights from theoretical computer science. We also conclude that an efficient description of such models allows for an efficient approximation of a \QAOA optimisation landscape. Further investigations into the complexity theoretical implications of this could provide insights into the link between the structure target spaces, the complexity of problems and the computational power of \QAOA, albeit we need to leave these questions to future research, as our primary goal in this paper is to establish the framework and derive direct practical utility.

\subsubsection{Sampling Approach}\label{sec:uniform_sampled_sampling}
Approximating $E\ab(F_1)$ with an efficient theoretical model of the target space is obviously the preferred way to address a given problem, if feasible analytically. However, it is also possible to determine $\ab\{E\ab(\overline{\#_d (k)})\}_{d=0}^n$ and $\ab\{E\ab(\overline{\#_{d_1} (k) \#_{d_2} (k)})\}_{d_1 , d_2 = 0}^n$ empirically by sampling from a set of problem instances. Above, we described the problem instances of the uniform sampling example by giving a sampling routine for a random target space. Recall that this defines problem instances as random variables. If we now want to determine the expected values and covariance matrix empirically, we need a concrete sample of instances. Which in this example will be a set of realisations of problem random variables. 

In our experiments, we sampled 500 random instances for each state space dimension from $8 \leq n \leq 11$. The upper bound of $n = 11$ keeps computational cost of explicitly evaluating $F_1$ at a reasonable level, whereas the lower bound $n = 8$ ensures a sufficiently large state space. For the scope of this paper, this range is sufficient to show the effects of increasing the state space dimension. Every instance has a target set of $2^{n-1}$ states, making it scale with the size of the whole state space. As defined above, $\overline{\#_d (k)} = \frac{1}{\ab|T|}\sum_{k \in T} \#_d (k)$. Let $\mathcal{T}$ be the set of target sets of all sampled instances. Then, empirically determining $E\ab(\overline{\#_d (k)})$ means calculating the mean of $\overline{\#_d (k)}$ over the set $\mathcal{T}$ of all sampled target sets. To clarify over which set a mean is calculated, we introduce the notation $\overline{f(x)}^A = \frac{1}{|A|} \sum_{x \in A} f(x)$. Then,
\begin{displaymath}
\overline{\overline{\#_d (k)}^T}^\mathcal{T} = \frac{1}{|\mathcal{T}|}\sum_{T \in \mathcal{T}} \frac{1}{|T|} \sum_{k \in T} \#_d (k).
\end{displaymath}
The same holds for $E\ab(\overline{\#_{d_1}(k) \#_{d_2} (k)})$, with 
\begin{displaymath}
    \overline{\overline{\#_{d_1} (k) \#_{d_2}(k)}^T}^\mathcal{T} = \frac{1}{|\mathcal{T}|}\sum_{T \in \mathcal{T}} \frac{1}{|T|} \sum_{k \in T} \#_{d_1} (k) \#_{d_2} (k).
\end{displaymath}

\begin{figure*}
    \centering
    \includegraphics{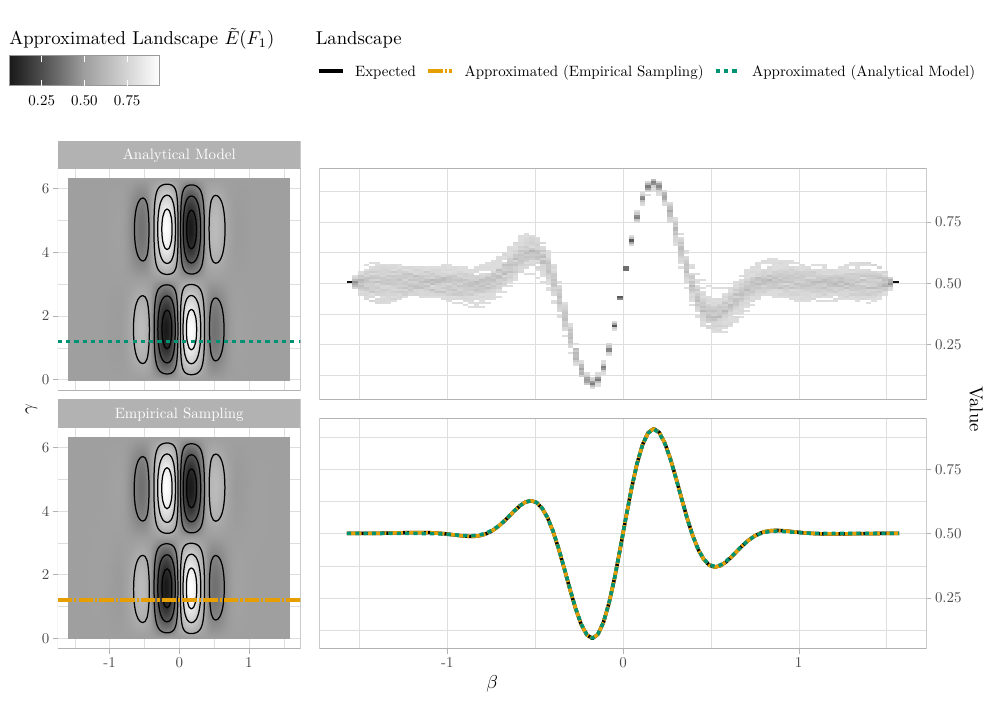}\vspace*{-1.5em}
    \caption{Comparison of analytic and sampling-based landscape approximation for the uniform sampling experiment.
    \emph{Left:} Landscapes for sampling methods. The difference between the analytically modelled and the empirically sampled landscape approximations is negligibly small.
    \emph{Top Right:} A two-dimensional heat-map \sqboxsquare{lfd3} (grey) that illustrates the vertical distribution of $F_1(\mangle, \pangle_0 = 1.2)$ values of 500 random instances sampled as described in \cref{sec:uniform_sampled_sampling}. The heat-map is calculated over a grid of $100 \times 100$ bins, where each column collects all $F_1$ values in the corresponding $\mangle$ interval and sums up to 1.
    \emph{Bottom Right:} Cross-section through the above landscape at $\pangle_{\text{c}} = 1.2$. The expected landscape estimated by the mean value of the data points displayed in the heat map above \sqbox{black} (black line) and two approximations of $F_1$, where $\ab\{\overline{\#_d (k)}\}_{d = 0}^n$ and $\ab\{(\overline{\#_{d_1} (k)}, \overline{\#_{d_2} (k)})\}_{d_1, d_2 = 0}^n$ are determined empirically \sqbox{lfd2} (orange line) and modelled theoretically \sqbox{lfd4} (green line) are shown. Again, exact and approximate results are in excellent agreement, with little variation.}
    \label{fig:uniform_analytic}
\end{figure*}

\subsubsection{Comparison}
As we can see in \cref{fig:uniform_analytic}, both approaches lead to approximations that fit the expected value of $F_1$ exceptionally well. In all our experiments we used the sample mean (over all instances) as an estimator for the expected landscape $E\ab(F_1)$. We calculated the mean of $F_1$ by evaluating $F_1\ab(\pangle, \mangle)$ for every instance at 100 sample points for $0 \leq \mangle \leq \pi$ at the cross section for $\pangle = 1.2$. This shows the versatility of our approximation: Theoretical models can be used to get excellent results. This is not always feasible for more complex problems where the inherent structural properties of the target space are not as straightforward. However, we also showed that an empirical approach leads to equally excellent results in such cases. We therefore argue that our approximation theorem can be a useful tool that is equally well suited for theoretical considerations and empirical analyses. 

\subsection{Clustered Sampling}
\begin{figure*}[htbp]
    \includegraphics{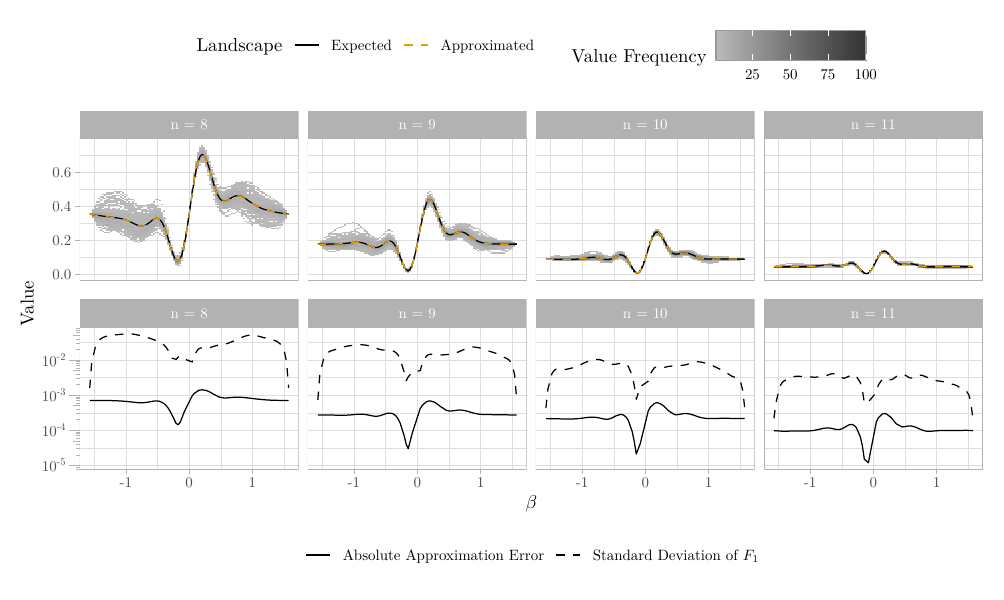}\vspace*{-2em}
    \caption{Landscape cross-sections and approximation quality for clustered sampling for different state space dimensions 
    in columns. \emph{Top:} Two-dimensional heat-map of sampled values for $F_1 (\mangle, \pangle_{c})$ at $\pangle_{c} = 1.2$ and across $0 \leq \mangle \leq \pi$, overlaid with the expected landscape estimated by the mean value of the sampled data points $\overline{F_1}$ \sqbox{black} (black line) and the approximated expectation $\tilde{E}\ab(F_1)$ \sqbox{lfd2} (orange line). \emph{Bottom:} Absolute approximation error (solid line) and standard deviation of $F_1$ (dashed line). The data demonstrate a small approximation error that lies strictly and considerably below the standard deviation of $F_1$.}
    \label{fig:clustered_dims}
\end{figure*}

The uniform sampling example was chosen deliberately trivial to showcase the analytical approach. As a consequence, one important aspect of more realistic problems remains unaddressed: Owing to the uniform sampling process, the states in $T$ are independent of each other. This is, in general, not the case for real-world problems. \SAT instances, for example, often have clustered solution spaces. Industrial \SAT instances exhibit community-like structures in their solution spaces~\cite{AnsoteguiGL12}. But contrary to intuition, even uniformly sampled \SAT instances possess solution structures~\cite{PariLYQ04}. Inspired by this observation, we conduct further experiments with a random clustered sampling process: Three cluster seeds are sampled uniform at random from the complete state space. Then, per seed, 30 new states are added to the target set. Each state is reached by a random walk starting from its corresponding seed by flipping a random bit each step, with a step probability of \sfrac{1}{2}. Note that reaching a basis state $\ket|k>$ from another basis state $\ket|z>$, with $k, z \in \mathds{F}_2^n$ means that qubit $l$ was flipped if and only if $k_l \neq z_l$. The process of empirically determining $\ab\{E\ab(\overline{\#_d (k)})\}_{d=0}^n$ and $\ab\{E\ab(\overline{\#_{d_1} (k) \#_{d_2} (k)})\}_{d_1, d_2 = 0}^n$ stays the same as in the uniform example above.

\Cref{fig:clustered_dims} shows the empirical estimate of $E\ab(F_1)$ compared to the actual mean value $\overline{F_1}$ and the introduced absolute error. As can be seen, we clearly introduce some error by approximating. But although the variance of $F_1$ drops significantly with higher dimensions, the absolute error never surpasses the standard deviation of $F_1$. Also note how the whole co-domain of $F_1$ is compressed. This is a direct consequence of the scaling coefficient \(\ab|T|/2^n\) in \cref{eq:F_1} as we chose to generate fixed dimension independent sized target sets. Recall that $E\ab(F_1)$ describes the result probability of sampling a solution state from $\ket|\pangle, \mangle>_1$, which unveils an interesting relation between the size of the solution space and the solution sample probability.

\subsection{Boolean Satisfiability}
\label{sec:sat_example}

\begin{figure}[htbp] \includegraphics{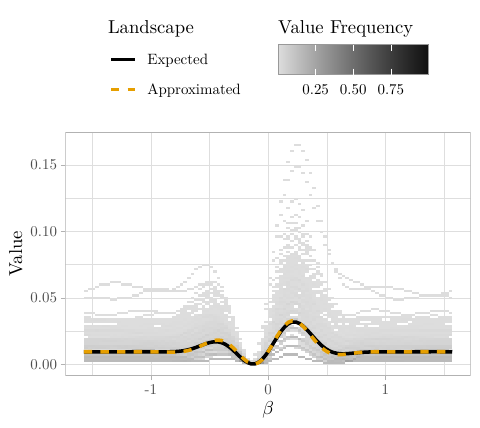}\vspace*{-1em}
    \caption{Approximate expected value $\tilde{E}(F_1)$ of $F_1\ab(\mangle, \pangle)$ for a random \SAT instance perfectly fits the actual expected landscape $E\ab(F_1)$ over all sampled instances, even though \SAT induces a higher variance in the values of $F_1$. This is illustrated by a cross section at $\pangle = \pangle_{0} = 1.2$. 
    The two-dimensional heat-map \sqboxsquare{lfd3} (gray) illustrates the vertical distribution of $F_1(\mangle, \pangle_0)$ over a grid of $100 \times 100$ bins, where each column collects all $F_1$ values in the corresponding $\mangle$ interval and sums up to 1. This histogram is overlaid by the expected Landscape $E\ab(F_1(\mangle, \pangle_{0}))$ \sqbox{black} (black, solid) and the approximated mean value $\tilde{E}(F_1\ab(\mangle, \pangle_{0}))$ \sqbox{lfd2} (orange, dashed) obtained from the approximation theorem. All three quantities are in very good agreement.}
    \label{fig:sat_cross_section}
\end{figure}

After two artificially constructed sampling examples, it is time to consider our first real problem. While \SAT is one of
the cornerstones of \NP-complete problems in theoretical computer science, it also has considerable practical impact~\cite{Ganesh:2020},
and has also been intensively studied in the 
quantum 
domain~\cite{Gabor:2019,Sax:2020,Krueger:2020,Willsch:2022,Zielinski:2023}.

Most importantly, it is known to often exhibit structured solution spaces~\cite{achlioptas2011solution}, and is a natural combinatorial decision problem, which makes it perfectly suited for our framework. 

Note that in both above sampling examples, we only discussed the target set and completely disregarded any actual realisation of the \QAOA circuit in our analysis. However, if the complete target set were known, it would be trivial to construct a constraint Hamiltonian $C = \prod_{k \in T} \ketbra|k><k|$ by simply projecting onto every target state. This is obviously not an adequate approach for hard computational problems, as full knowledge about the solution space cannot be expected.
Unfortunately, explicitly projecting onto every state in $T$ violates this requirement for difficult problems. We therefore show now that we still can apply our approximation theorem even if this requirement is added to the picture.

Let $\mathcal{C} = \ab\{c_i\}_{i=1}^m$  a Boolean formula in conjunctive normal form, with $c_i$ being the $i$-th clause. The \SAT problem asks whether there exists a Boolean variable assignment such that all clauses are satisfied. Each clause is a disjunction of literals. A literal is a possibly negated Boolean variable. Let's consider the clause $c = x_1 \vee x_3 \vee \overline{x_8}$. Then $c$ is satisfied by all possible assignments except for the characteristic unsatisfying assignment $\ab[x_1 \mapsto 0, x_3 \mapsto 0, x_8 \mapsto 1]$. Let $\overline{C_i}$ be a projector onto this characteristic unsatisfying assignment of the clause $c_i$. Given $\{\overline{C_i}\}_{i=0}^m$ we can construct a constraint Hamiltonian $C = \prod_{i = 0}^m (\mathds{1} - \overline{C_i})$ that projects onto all satisfying assignments of $\mathcal{C}$. This allows us to just focus on the solution space, and set the stage for 
using the approximation theorem.

For our experiments, we generated 500 random satisfiable \SAT instances for each $n \in \{8, 9, 10, 11\}$, where the number of variables directly translates to the dimension ($|V| = n$) and the number of clauses is always $|C| = 4n$. Each clause has three literals that are negated with probability \sfrac{1}{2}. Glucose 4.2~\cite{Audemart:2009} was used to enumerate all satisfying assignments for each instance. In contrast to the previous sampling examples, where we always sampled a target set of fixed size $|T|$, the size of $T$ depends on the instance and its number of satisfying assignments in this scenario. Therefore, we also have to determine $|T|$ empirically.

In \cref{fig:sat_cross_section} we see once again how well our approximation fits the real mean values for $F_1$. Note that compared to the examples above, there seems to be a unusual amount of variance on the y-axis. This is a result of different target set sizes $|T|$ across instances, which causes global vertical shifts of function $F_1$. If one is just interested in the overall structure of $F_1$ one could consider ignoring the scaling factor induced by $\ab|T|$ altogether. 

\subsection{\kclique}
\label{sec:kclique}
An efficient in-place projection like the \SAT constraint Hamiltonian can not always be expected. So what if ancilla bits are needed? In this example, we demonstrate an approach following \cref{thm:ancilla_invariance} by showing how to construct an ancilla register independent of the constraint Hamiltonian for \kclique. 

\begin{definition}
    Given a graph $G$ and $k \in \mathds{N}$, then the decision problem whether $G$ has a clique (\ie, a fully connected sub-graph) of size $k$ is called the \kclique Problem.
\end{definition}

To solve \kclique with \QAOA, we first need to define the state space. Let $G = (V, E)$ be a graph with $n$ vertices, then we have a $n$-dimensional state space as we map each vertex to a specific qubit. A basis state $\ket|z>$ with $z \in \mathds{F}_2^n$ marks vertices such that a vertex $v_i$ is marked iff $z_i = 1$. Then, $\ket|z>$ represents a valid solution to \kclique iff $K = \{v_i \in V \mid z_i = 1\}$ is a clique in G and $|K| = k$. 

Let $\overline{G}$ be the complement of $G$ and $(i, j)$ an edge in $\overline{G}$. Then,
\begin{equation}
\begin{gathered}
    \Ccliques = \prod_{(i,j) \in \overline{G}} \ab(\mathds{1} - P_{ij}) \\
    \text{with} \quad P_{ij} = \mathds{1}_{0, i - 1} \otimes \ketbra|1><1| \otimes \mathds{1}_{i+1,j-1} \otimes \ketbra|1><1| \otimes \mathds{j}_{j+1,n}
\end{gathered}
\end{equation}
projects onto all cliques in $G$. However, this also includes cliques of sizes different from $k$, such as trivial cliques as single vertices and edges. A second step is needed to filter out the $k$-cliques. For this we have to define an unitary operator $D_H$ that calculates the Hamming weight of $\ket|z>$. This is done by writing $d_H(z)$ on an ancilla register but note that $\ket|z>\ket|y> \mapsto \ket|z>\ket|d_H (z)>$ is only invertible for a fixed $y$ (\eg, $y=0$), which conflicts with $D_H$ being unitary. Thus, we define $D_H: \mathcal{B}^{\otimes n} \otimes \mathcal{B}^{\otimes \lceil \lb(n) \rceil} \to \mathcal{B}^{\otimes n} \otimes \mathcal{B}^{\otimes \lceil \lb(n) \rceil}$ as $D_H : \ket|z>\ket|y> \mapsto \ket|z>\ket|y + d_H (z) \mod 2^m>$ with $m \coloneqq \lceil\lb(n)\rceil + 1$. Since $\mathds{N}/m\mathds{N}$ forms a group under addition, there exists a unique inverse element for each $d_H(z)$ which ensures $D_H$ to be invertible. Our construction is illustrated in \cref{fig:kclique_DH}.

With $D_H$ we can construct a second constraint Hamiltonian $\Cdhk = D_H^\dagger (\mathds{1}^{\otimes n} \otimes \ketbra|k><k|) D_H$ that projects onto the space spanned by $\{\ket|z>\ket|y> \mid y + d_H (z) \mod 2^m = k\}$. Therefore, if the ancilla register is initialised to $\ket|\bm{0}>$ the application of $\Cdhk$ projects onto states with Hamming distance $k$. Applying both projectors results in
\begin{equation}
\begin{gathered}
    C \coloneqq \Cdhk \ab(\Ccliques \otimes \mathds{1}^{\otimes \lceil \lb(n) \rceil}) \quad \text{with}\\
    C\ket|z>\ket|\bm{0}> = \begin{cases}
        \ket|z>\ket|\bm{0}>  &\quad z \text{ is k-clique in } G \\
        0                        &\quad \text{otherwise}
    \end{cases}
\end{gathered}
\end{equation}

This results in a \QAOA state $\ket|\mangle, \pangle>_1 = e^{-i \mangle X^{\otimes n}} e^{-i \pangle C} \frac{1}{\sqrt{2^n}} \ket|+>^{\otimes n}\ket|\bm{0}>$ after the phase separation introduced by $e^{-i \pangle C}$ the state evolves to $\frac{1}{\sqrt{2^n}} e^{-i \mangle X^{\otimes n}} \ab( e^{-i \pangle} \sum_{z \in K} \ket|z>\ket|\bm{0}> + \sum_{z \notin K} \ket|z>\ket|\bm{0}>)$ finally, after factoring out the ancilla register and applying the mixer we end up with
\begin{displaymath}
\begin{aligned}
    \frac{1}{\sqrt{2^n}} \Biggl(&e^{-i \pangle} \sum_{z \in K} \sum_{x \in \mathds{F}_2^n} f(\mangle, z, x) \ket|x>  + \\ 
                                &\sum_{z \notin K} \sum_{x \in \mathds{F}_2^n} f(\mangle, z, x) \ket|x>\Biggr) \otimes \ket|\bm{0}>
\end{aligned}
\end{displaymath}
Note that this \QAOA circuit is invariant on the ancilla register, which 
allows us to trace the ancilla qubits out. Then, $\ket|\mangle, \pangle>_1$ 
is identical to \cref{eq:qaoa_state_1_mixer_applied}. The ancilla qubits have 
no influence on the distribution of Hamming weights. Therefore, they can be 
ignored and once again we only have to reason about the problem specific 
solution space.

\begin{figure*}[htbp]
    \centering
    \includegraphics{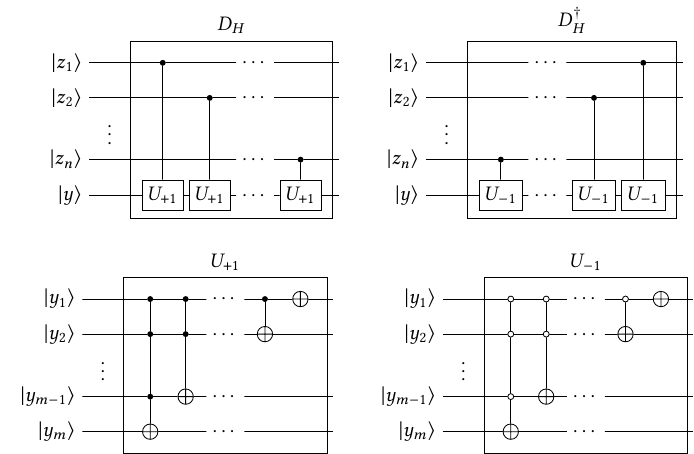}
    \caption{Circuits for Hamming weight computation. $D_H$ is used to construct the constraint Hamiltonian for the \kclique problem. It is used to check whether a potential clique is of size $k$ or not. For this $D_H$ adds the Hamming weight of register $\ket|\bm{z}>$ to another register $\ket|\bm{y}>$, while $D_H^\dagger$ subtracts it from $\ket|\bm{y}>$. State $\ket|\bm{y}>$ is initialised with $\bm{y} = 0$ and holds the Hamming weight of $\bm{z}$ encoded as a bit string, after the application of $D_H$. The $D_H$ gate adds the Hamming weight to $\ket|\bm{y}>$ by conditionally applying $U_{+1}$ to it for each qubit in $|\bm{z}>$, with $U_{+1}\ket|\bm{y}> = \ket|\bm{y} + 1>$. The inverse operation $D_H^\dagger$ subtracts the Hamming weight of $\bm{z}$ from $\ket|\bm{y}>$, by conditionally applying $U_{-1}$ for each qubit, with $U_{-1}\ket|\bm{y}> = \ket|\bm{y} - 1>$.}
    \label{fig:kclique_DH}
\end{figure*}

\subsection{One-Way Functions}
\label{sec:one_way_functions}

\Cref{thm:approx} is especially useful if we either have a theoretical understanding of the solution space of a problem or if the solution space can be efficiently sampled. This is the case for one-way functions: Instances can be easily generated by first picking a state from solution space and then generating the original input instance by applying the inverse of the one-way function, as it is easy to compute by definition. We will showcase this scenario with \qrfactoring.
\begin{definition}
    Given a value $x = qr$ with $q, r \in \mathds{P}$, \qrfactoring is the problem of computing $q$ and $r$ given $x$.
    \begin{equation}
        \label{eq:fqr}
        f_{qr} : \ab\{qr \mid q, r \in \mathds{P}\} \to \mathds{P} \times \mathds{P} ,\quad f_{qr} : qr \mapsto (q,r)
    \end{equation}
\end{definition}
Obviously the inverse of \cref{eq:fqr} is simply the integer multiplication $f_{qr}^{-1}(q,r) = q \cdot r$. So, given a pre-computed set of primes $\mathcal{P} \subset \mathds{P}$, we can efficiently sample a pool of \qrfactoring instances by $f_{qr}^{-1} \ab(\ab\{(q,r)_i\}_{i=1}^m)$ with $\ab(q,r)_i$ being sampled from $\mathcal{P} \times \mathcal{P}$ at uniform random for all $0 \leq i \leq m$. In fact, as we argued in \cref{sec:kclique}, the target space $T$ without ancillary qubits fully suffices for our approximate analysis. For \qrfactoring, a solution $(q,r)$ can be represented in $T$ as follows: Let $x_{(2)}$ be the binary representation of $x \in \mathds{N}$, and let $z \in \mathds{F}_2^{k}$ with $m \leq n$. Then $p_n(z)$ is a padding of $z$ with $n - k$ leading zeros. Now, $\ab\{(q,r)_i\}_{i = 1}^m$ is mapped to $T$ by $(q,r) \mapsto p_n(q_{(2)}) \circ p_n(r_{(2)})$ with $n = \max\ab\{\lceil \lb(q) \rceil, \lceil \lb(r) \rceil\}$, where $\circ$ denotes concatenation of two bit strings.

With this mapping we performed a series of experiments for different target space dimensions $n \in \ab\{12, 14, 16, 18\}$, where we sampled $\ab\{(q,r)_i\}_{i = 1}^{100}$ from $\mathcal{P}_{\frac{n}{2}} \times \mathcal{P}_{\frac{n}{2}}$, with $\mathcal{P}_{\frac{n}{2}} \coloneqq \ab\{q \in \mathds{P} \mid q < 2^{\frac{n}{2}} \}$. Again, $\ab\{E(\overline{\#_d (k)})\}_{d=0}^n$ and $\ab\{E\ab(\overline{\#_{d_1} (k) \#_{d_2} (k)})\}_{d_1, d_2 = 0}^n$ were empirically determined as described above. The cardinality of the target set is $\ab|T| = 2$, as $f_{qr}^{-1}(q,r) = f_{qr}^{-1}(r,q)$ and thus for each $x = qr$, solution and target space each contain exactly two elements. \Cref{fig:qr_factoring_n18} shows the approximated and actual mean of $F_1$ as well as the absolute approximation error for $n = 18$. Although the margin for error apparently is extremely tight for \qrfactoring, our approximation nicely lies on top of the actual mean. Again, the absolute error remains below the standard deviation. 
As is likewise visible in \cref{fig:qr_factoring_n18}, \qrfactoring exhibits only small  deviations between instances in their optimisation landscapes. But even with this small margin of error, the approximation tightly fits the real distribution, and the error is bounded by the standard deviation of $F_1$.

\begin{figure}[htbp]
    \includegraphics[trim=0em 1.5em 0 1.5em]{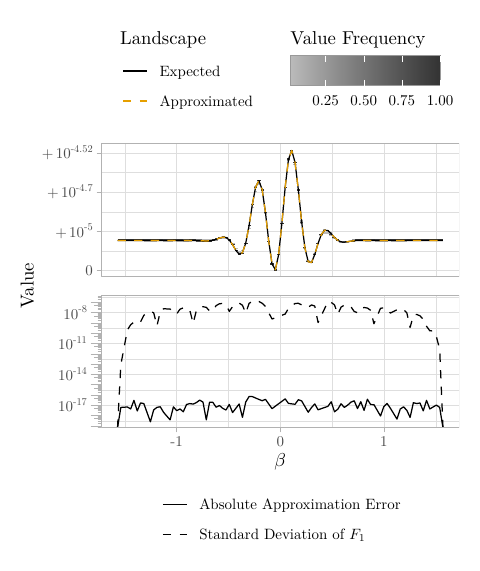}
    \caption{\emph{Top:} 2D histogram of the $F_1$ values for all instances for \(n=18\) from the \qrfactoring example, overlaid by  the approximated landscape $\tilde{E}(F_1)$ \sqbox{lfd2} (orange) and the expected landscape \(E(F_{1})\) \sqbox{black} (black). \emph{Bottom:} Absolute approximation error (solid line) compared to the standard deviation of $F_1$ (dashed line). The variance in $F_1$ is negligible, which makes the margin of error for approximations extremely tight. Our approximation provides excellent results even under such challenging conditions.}
    \label{fig:qr_factoring_n18}
\end{figure}

\section{Practical Utility}\label{sec:utility}
It is well known that the classical task of optimising parameters in \QAOA requires substantial amounts of computational effort: This aspect of the heuristic is \NP-hard~\cite{bittel_training_2021} even for classically tractable systems; likewise, every polynomial time algorithm is susceptible to instances for which the relative resulting error can be arbitrarily large, thus rendering approximate approaches likewise troublesome. The \QAOA quantum circuit, in addition, usually needs to be evaluated for many different sets of angles to gain the necessary information on the optimisation landscape that is required as input for the classical parameter optimisation routine, albeit efforts differ depending on the concrete choice~\cite{Periyasamy:2024,Thelen:2024,Singhal:2024}

\begin{figure*}[htbp]
    \centering
    \includegraphics{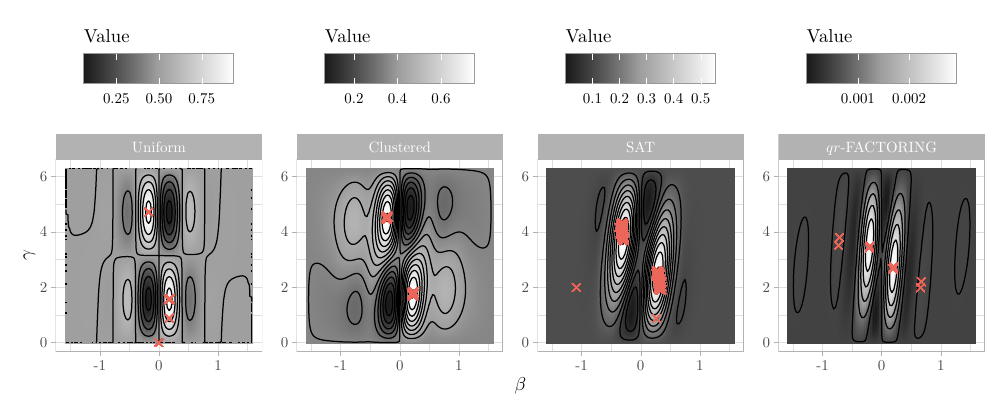}\vspace*{-1em}
   \caption{Each red cross marks a combination \((\mangle, \pangle\)) obtained 
   from the standard \QAOA algorithm for one single instance of the 
   subject problem. The parameter pairs are plotted on top the expected optimisation landscapes. As is visually apparent, optimal parameters for 
   different instances follow the \emph{problem-specific}, but \emph{instance-independent} patterns 
   expected from our considerations. Solutions for local optima
   of the instance-averaged landscape, as they arise for some \QAOA 
   executions, are typical artefacts expected from numerical 
   optimisation.}\label{fig:circopt_landscape}
\end{figure*}

Our approximation approach allows us to separate instance sampling from optimisation landscape sampling. Even at a fixed point $\ab(\mangle, \pangle)$, directly sampling $F_1 \ab(\mangle,\pangle)$ is intractable for
exact computation, as it contains a sum over exponentially many terms. Likewise, determining an expectation value $E\ab(F_1)$ requires to consider a substantial amount of points $\ab(\mangle, \pangle)$ and instances. However, given structural information about the target space in form of $\ab\{E\ab(\overline{\#_d (k)})\}_{d=0}^n$ and $\ab\{E\ab(\overline{\#_{d_1} (k) \#_{d_2} (k)})\}_{d_1, d_2 = 0}^n$, only a sum over linear many terms has to be evaluated to approximate $E\ab(F_1)$ efficiently at a point $\ab(\mangle, \pangle)$. The inputs $\ab\{E\ab(\overline{\#_d (k)})\}_{d=0}^n$ and $\ab\{E\ab(\overline{\#_{d_1} (k) \#_{d_2} (k)})\}_{d_1, d_2 = 0}^n$ of our approximation method can reasonably be obtained empirically by statistical methods, or even theoretically modelled as demonstrated in \cref{sec:uniform_sampled}. Additionally, \cref{thm:approx} takes $E\ab(\ab|T|)$ as an input. This quantity is in fact the expected value of the solution of the corresponding counting problem. Counting problems and approximate counting have been studied extensively by the theoretical computer science community since the early days \cite{stockmeyer1983complexity,dyer2003approximate,wei2005new}. In general even approximate counting is a hard problem in most cases. Luckily, $E\ab(\ab|T|)$ only acts as a scaling factor. In most cases where one is only interested in the overall structure of the optimisation landscape, \ie local optima, plateaus, etc. the global scaling factor can be discarded. Thus, an accurate approximation of $E\ab(\ab|T|)$ is not as important for the application of out main theorem.

The above mentioned separation of instance and object landscape sampling enables a different approach to \QAOA that does not require an unbounded iteration involving the quantum circuit. It starts by sampling (partial) target spaces of random instances, from which the structural metrics needed as input for \cref{thm:approx} are gathered. Then, a classical optimiser is utilised to determine the optimal \QAOA angles with regard to the approximate optimisation landscape, resulting from the previously sampled target spaces.
Because we used the expected landscape of a random instance of the problem at hand, the found parameters apply to the complete problem. Thus, parameter optimisation only needs to be performed once per problem to find one single set of parameters for
all instances.
Finally, the resulting angles are used to initialise a \QAOA circuit from which potential problem solutions will be sampled on quantum hardware. Following this approach, only the final sampling step needs to be performed on real quantum hardware. The standard \QAOA procedure is in fact more of a heuristic than a quantum algorithm, where we have to optimise a quantum circuit for each instance. In our approach to \QAOA we optimise the parameters on a error bounded approximation of the expected landscape. Thus, we end up with one circuit for all problem instances. This mathematically sound splitting of computation into a problem-global and instance-specific phase significantly moves \QAOA towards the realm of true quantum algorithms, in stark contrast to empirically motivated heuristics. Recall \cref{fig:application_overview}, where we illustrated the difference between both approaches.

\begin{figure}[htbp]
    \includegraphics{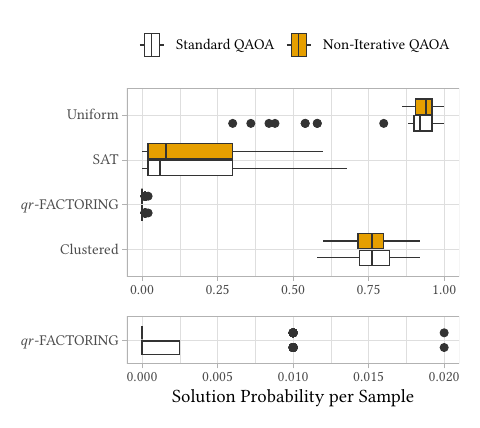}\vspace*{-1em}
    \caption{Probability of states sampled from the \QAOA circuits to encode valid solution for the various subject problems considered in this paper with a standard \QAOA approach (individual parameter optimisation for each instance), and with our non-iterative schema using a-priori
    parameters obtained from the approximated landscape for all instances.}\label{fig:circopt_vs_preopt}
\end{figure}

We explored this approach empirically (using numerical simulations; see \cref{sec:repro} for details about our reproduction package that allows researchers to directly employ our approach or inspect our implementation) for the examples introduced above in \cref{sec:application}. For every example, we randomly generated 50 instances to be solved with both \QAOA methods: standard \QAOA and \ourQAOA. The approximate optimisation landscape for the \ourQAOA was calculated based on the dataset generated earlier in \cref{sec:application}. After  optimisation, we visually verified each set of parameters for \ourQAOA resides on one of the two main extrema of the landscape. After that, the parameters were used to create a \QAOA circuit for each problem instance. This circuit then was compared with a \QAOA circuit trained on the actual instance. From both circuits, we sampled 50 potential solutions per instance (100 states were sampled per instance for \qrfactoring to more accurately capture the extremely low maximal possible success probability). \cref{fig:landscapes} shows the resulting approximate landscapes.

As for practical performance, we can see from \cref{fig:circopt_landscape} that parameters obtained by standard \QAOA are mostly located around extrema of the approximate optimisation landscapes. However, the optimisation apparently also often ends in local maxima, especially for \qrfactoring. This suggests matching success rates for standard and \ourQAOA. Indeed, we can see in \cref{fig:circopt_vs_preopt} that states sampled from \ourQAOA are valid solutions with identical (or slightly better) probability than states sampled from standard \QAOA. 
Hard combinatorial constraint satisfaction problems are known to have phase transitions between trivially under- and over-constrained instances, and the actually hard instances reside in the parameter region of this very phase transition. For \SAT, the level of under- or over-constrainedness is measured by the coefficient of variables to clauses $\alpha = \sfrac{|C|}{|V|}$, where $|C|$ and $|V|$ denote the number of clauses and variables, respectively. The phase transition of \SAT happens at $\alpha \approx 4$. For our comparison, we generated 50 instances for each value of $\alpha = 2$, $\alpha = 4$ and $\alpha = 6$. As before, the \ourQAOA approach is on par with standard \QAOA for under- ($\alpha = 2$) and over-constrained ($\alpha = 6$) \SAT instances, as well as for hard \SAT instances right in the phase transition at $\alpha = 4$. This underlines the utility of our approach not only for \enquote{average}, but also hard instances. \cref{fig:sat_resolved} visualises the outcome distribution in detail.

\begin{figure}[htbp]
    \includegraphics{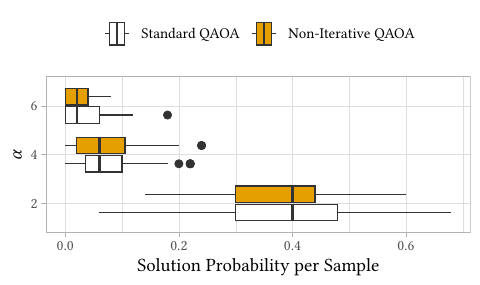}
    \caption{Probability of states sampled from the \QAOA circuits to encode a valid solution for \SAT instances of different hardness $\alpha$. The \ourQAOA variant is on par with standard \QAOA for trivially under- ($\alpha = 2$) and over-constrained ($\alpha = 6$) instances, as well as for usually hard instances in-between at $\alpha = 4$ (a comparison between standard and \ourQAOA over all \SAT instances is included in \cref{fig:circopt_vs_preopt}).}
\label{fig:sat_resolved}
\end{figure}

Judging from the above evaluation, our \ourQAOA algorithm at least 
matches standard \QAOA success probabilities, and 
occasionally shows slight advantage. However, \QAOA involves 
classical parameter optimisation for every instance, while our 
approach is more resource efficient: Only a constant number of 
quantum circuit evaluations instead of sampling the circuit in every 
iteration of the classical optimisation loop is required, and 
we avoid instance-specific classical optimisation altogether:
One single up-front classical optimisation process is required 
to infer optimal parameters from the approximated landscape per 
problem. Costly circuit evaluations in the quantum-classical 
iteration of \QAOA are replaced by classical instance sampling in 
\ourQAOA: The quantum resource is only used for instance 
optimisation, not for preparatory work.

\section{Discussion}
\label{sec:discussion}
\begin{figure*}[htbp]
    \centering
    \includegraphics{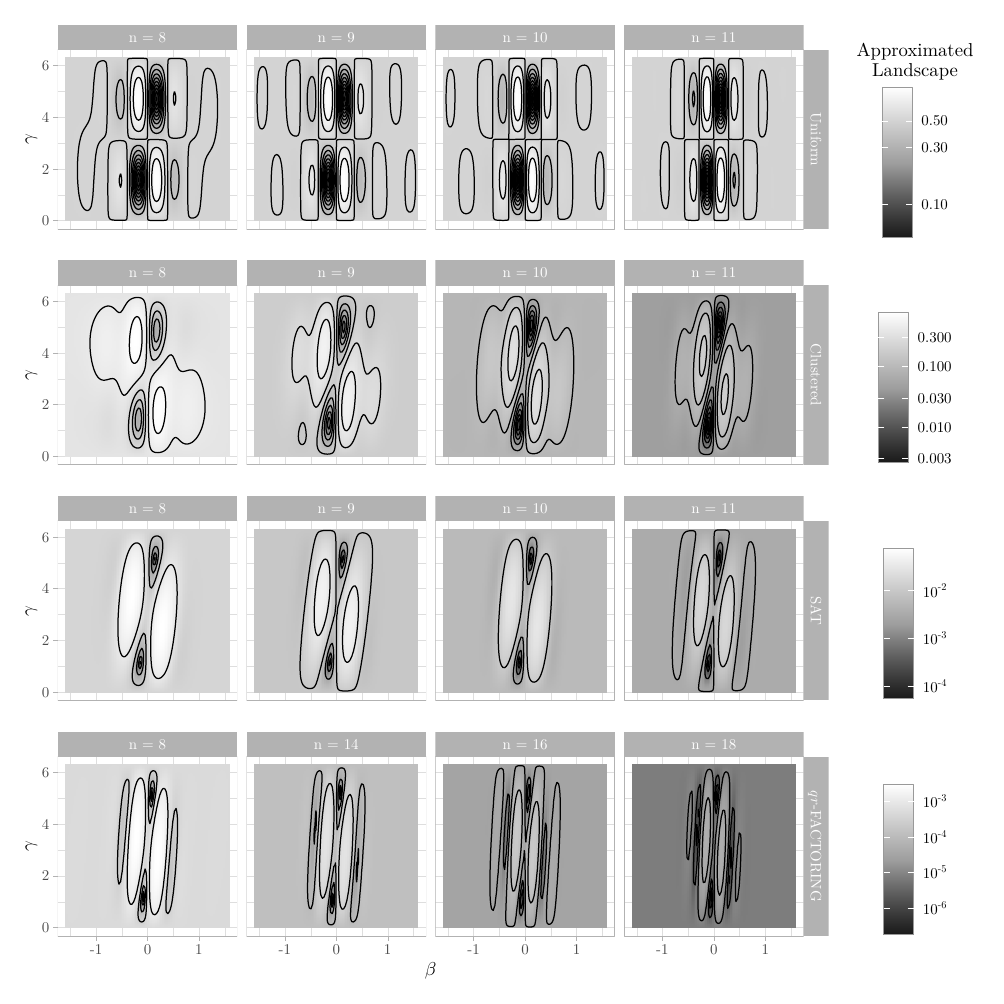}
    \caption{Overview of all approximated landscapes for the considered subject problems (100 random instances per problem and dimension). Each panel shows the 
    approximated optimisation landscape 
    \(\tilde{E}(F_{1}(\mangle, \pangle))\) obtained by \cref{thm:approx} as surrogate for the expected value of \(F_{1}(\beta, \gamma)\). Macroscopic similarities do not only arise between instances of different dimensions (left to right), but also to a certain extent across problems (top to bottom).}
    \label{fig:landscapes}
\end{figure*}
We have considered \QAOA from a mixed perspective comprising computer science and physics, and could not only establish a new 
approximation theorem for the optimisation landscape, but have also 
suggested an algorithm that translates these insights directly into useful advantages. The evaluation of both aspects in the preceding 
sections has shown that either improves upon the state of the art
in various ways.

However, our results also raise new questions.
Obvious issues include how to extend the approach
to \QAOA depths larger than one;  possible gains in post-\NISQ
systems; a more comprehensive evaluation on specific
industrial problems and larger instance sizes; and the
robustness to noise. Likewise, the question of how computational 
power that eventually arises from our algorithm is distributed 
between classical and quantum resources, and what impact this 
distribution has on the properties like runtime or solution 
quality, remains open. We need to leave addressing these
questions to future research.

For all subject problems analysed above in \cref{sec:application} 
our approximation matches the expected landscape extremely well, 
with negligibly small errors. Given that we could show the absolute 
approximation error to be bounded by $\sqrt{\Var\ab(|T|) 
\Var\ab(\overline{|c_k|^2})}$, it follows that for all problems with 
constant solution space sizes (\ie, $\Var\ab(|T|) = 0$), our 
approximation actually delivers exact results. It should be noted that this  
requires knowledge of exact values for 
$\ab\{E\ab(\overline{\#_d (k)})\}_{d=0}^n$, $\ab\{E\ab(\overline{\#_{d_1} (k) \#_{d_2} (k)})\}_{d_1, d_2 = 0}^n$ and $E\ab(\ab|T|)$, 
which is usually not practically achievable, especially for 
empirical sampling. This explains the approximation errors observed 
for all subject problems with fixed target sizes above---except for \SAT. This 
observation naturally raises questions about the relationship 
between the estimation accuracy of the target space structures and 
the resulting approximation error. Here, a trade-off between sample 
size and approximation quality is to be expected and should be 
analysed in future work. An answer to this question also could give 
interesting insight into the variance of optimisation landscapes 
between individual problem instances.

Regarding the scalability of out approach, we note that given the structural information about the target space the landscape approximation according to \cref{thm:approx} results in a evaluation of $\mathcal{O}\ab(n^2)$ terms, where $n$ is the number of qubits. Therefore, the landscape approximation itself scales efficiently with the number of qubits, obtaining the needed structural solution space information on the other hand depends on the chosen method. We demonstrated  that for this both analytical and empirical approaches can be taken. In the first case the efficiency depends on the theoretical model and the closed from or approximation arising from it. For the latter, it would be interesting to investigate the approximation robustness of \ref{thm:approx} under limited sampling. Another empirical approach where our framework could proof beneficial would be to deploy it in an online learning fashion, where the needed solution samples are gathered on the way. Although, out of scope and unaddressed in this paper, this could provide an interesting path investigate in future work. One could also be curious about the feasibility of our proposed landscape approximation for deeper circuits with an increasing number of layers. With more layers \QAOA \emph{sees} more of the graph, while it generally does not see the complete graph with $p < \log n$ \cite{farhi2020quantum,farhi2020quantumTypcal}. With \cref{thm:approx} operating on global structures, we argue that these global structures should even become more important for deeper \QAOA circuits scanning the complete qubit graph. 
 
Our optimisation landscape approximation is solely based on structural information about a problem. This allowed us to prove the existence of conjectured instance invariants that stem from observed effects like parameter clustering~\cite{brandao2018fixed,streif2020training}. The approximation allows us to obtain a smooth representation of the complete optimisation landscape. In \cref{fig:landscapes}, we collect such landscapes for the subject problems discussed in \cref{sec:application} for different dimensions of the target space $T$. This provides us with some interesting observations: On the one hand, it can be seen that macroscopic similarities exist between all problems. They share very similar high level features (with some problem specific differentiation). These macroscopic features are also present for structure-less target sets like in the case of the uniform sampling example. We thus conjecture that these features are truly problem independent and could be explained by either just the approximation theorem presented in this work, or in combination with straight forward and minimalistic examples like uniform sampled target sets. Additional structure causes some displacement and skewing. This structure can be plausibly explained by local effects like membership of a state to the target set depending on the membership of other neighbouring states. As we can additionally observe from \cref{fig:landscapes}, simply increasing the state space dimension (\ie, going from left to right on the panels for a specific problem) seems to compress the landscape along $\mangle = 0$ along the $\mangle$-axis. This effect is independent of the subject problem (\ie, does not change when traversing the panels from top to bottom in the plot). As it can also be observed in the uniform sampling example, it must be independent of structural properties. While further examining these effects could provide valuable insight into the behaviour of \QAOA optimisation landscapes, we need to leave the efforts required for such an investigation to future work, as they go beyond the scope of this paper. The approximation theorem devised in this paper will likely serve as an analytical basis for such considerations.

Focusing on two-level Hamiltonians and \QAOA circuits keeps the approximation framework simple and straightforward to use. We thus provide a well defined base framework to analyse the essential connection of problem structure and \QAOA landscapes. At the same time, it is extendable to accommodate, for instance, Hamiltonians with more complex eigenspectra or \QAOA circuits with more layers, albeit we also need to leave such efforts
to future work. However, it seems pertinent to note that the landscapes for the subject problems considered in this work match empirical results with higher level eigenspectra and deeper \QAOA circuits in their macroscopic
structure~\cite{pelofske2024short,streif2020training}. 
The same holds true for Ising model \QAOA landscapes based
on an instance-specific analytical derivation~\cite{Ozaeta:2022}.

We believe our work lays a foundation to map significant insights 
from classical theoretical computer science to quantum approaches 
along the lines of \QAOA. In the classical case, the analysis of 
solution space structures is well 
established~\cite{AnsoteguiGL12,PariLYQ04,hogg1996refining,achlioptas2011solution}. With the approximation theorem (and the
techniques introduced to derive it), we present a handle to
apply this knowledge to quantum optimisation. This provides an opportunity for future progress in understanding and 
utilising the class of algorithms initiated by \QAOA.

\section{Conclusion}\label{sec:conclusion}
Our new perspective on \QAOA from a mixed physics and theoretical 
computer science point of view allowed us to accurately approximate 
the \QAOA optimisation landscape based on the inherent structural 
properties of a problem, and prove a long-standing set of hypotheses 
about relationships between optimal \QAOA parameters for 
different instances of a problem. By directly linking the solution 
space structure of problems to their \QAOA optimisation landscapes, 
we could devise an approximation theorem that does not only provide 
structural insights, but also impacts the practical use of \QAOA.

Based on our results, we have constructed a
\ourQAOA algorithm that is more resource-efficient than
standard \QAOA for various measures. Our evaluation on five different scenarios
and subject problems has shown excellent agreement from
the theoretical and practical point of view, and provides
concrete advantages over standard \QAOA. Our new perspective
on \QAOA opens various possibilities for future research
to understand unsolved issues about quantum optimisation.

\newcommand{\TK}{\censor{TK}\xspace}
\newcommand{\WM}{\censor{WM}\xspace}
\newcommand{\programme}{\blackout{German Federal Ministry of
    Education and Research (BMBF), funding program
    \enquote{Quantum Technologies---from Basic Research to Market}}}
\newcommand{\grantoth}{\censor{\#13NI6092}}
\newcommand{\hta}{\censor{High-Tech Agenda Bavaria}}

\vspace*{1em}\textbf{Acknowledgements:} We thank Jonathan Reyersbach for 
many active discussions that helped shape the mathematical
exposition of our ideas, and acknowledge important comments from 
Lukas Schmidbauer, Simon Thelen, and Maja Franz to ascertain the 
correctness of our main theorem. We are grateful for funding from 
\programme, grant \grantoth{}. \WM also acknowledges support
by the \hta.

\begin{appendix}
\section{Reproduction Package}\label{sec:repro} 
    
To make our experiments reproducible by other researchers~\cite{reprobook:2019}, we provide the complete source 
code for the calculations performed in the paper, in form of a long-term 
stable~\cite{Mauerer:2022} \repro (link in PDF; a
\doirepro version \cite{reprPackage} 
is also available). We ascertain that the package is fully
self-contained, and does not rely on resources that may eventually 
vanish from public access. In particular, we provide code that is
easily extensible to more subject problems than considered
in the paper, and include routines to (a) perform target space 
sampling; (b) implement \ourQAOA; (c) perform numerical simulations 
to explore performance and compare with standard \QAOA.
Any raw data obtained from our simulations, together
with a complete pipeline to perform the evaluation and visualisation
presented in the paper, are included.
\end{appendix}

\bibliographystyle{plainnat}
\bibliography{main}

\end{document}